\documentclass{article}

\usepackage[utf8]{inputenc}
\usepackage[braket]{qcircuit}
\usepackage{graphicx}
\usepackage{epstopdf}
\usepackage{authblk}
\usepackage{dcolumn}
\usepackage{bm}
\usepackage{mathrsfs}
\usepackage{xcolor}
\usepackage{setspace}
\usepackage[margin=1.1in]{geometry}
\usepackage{hyperref,appendix}
\usepackage{amsmath,amssymb}
\usepackage{makecell,soul}

\begin{document}


\title{Entanglement-assisted  concatenated quantum codes}

\author[1]{Jihao Fan}
\author[2]{Jun Li}
\author[1]{Yongbin Zhou}
\author[3]{Min-Hsiu Hsieh}
\author[4]{H. Vincent Poor}
\affil[1]{School of Cyber Science and   Engineering, Nanjing University of Science and Technology, Nanjing 210094, China}
\affil[2]{School of Electronic and Optical Engineering, Nanjing University of Science and Technology, Nanjing 210094, China}
\affil[3]{Hon-Hai Quantum Computing Research Center, Taipei, Taiwan}
\affil[4]{Department of Electrical and Computer
Engineering, Princeton University,   NJ 08544, USA}

\maketitle
\begin{abstract}
Entanglement-assisted   concatenated quantum codes (EACQCs), constructed
by concatenating two quantum codes,  are proposed. These EACQCs   show  several advantages over the standard
 concatenated  quantum codes (CQCs). Several families of EACQCs that, unlike standard CQCs,
  can beat  the nondegenerate  Hamming bound for  entanglement-assisted quantum  error correction   codes (EAQECCs) are derived. Further,  a number  of EACQCs with  better parameters
 than the best known standard   quantum error correction   codes  (QECCs) and EAQECCs are also derived.  In particular, several  catalytic EACQCs with better parameters than the  best known QECCs of the same length and net transmission are constructed.  Furthermore, each  catalytic EACQC consumes only one or two ebits. It is also shown   that EACQCs make  entanglement-assisted   quantum communication possible even if the ebits are noisy.  Finally, it is shown that   EACQCs  can outperform CQCs in   entanglement fidelity   over depolarizing channels  if the ebits are less noisy than the qubits.    Moreover, the threshold    error probability of EACQCs is larger  than that of CQCs when the error probability of ebits is sufficiently lower than that of qubits. Therefore EACQCs are not only competitive in quantum communication but also applicable in fault-tolerant quantum computation.
\end{abstract}


\maketitle

\newtheorem{definition}{Definition}
\newtheorem{theorem}{Theorem}
\newtheorem{result}{Result}
\newtheorem{lemma}{Lemma}
\newtheorem{corollary}{Corollary}
\newtheorem{example}{Example}
\newtheorem{proposition}{Proposition}
\newenvironment{proof}{{\noindent\it Proof:}}{\hfill $\square$\par}

\section{Introduction}
Quantum   error  correction codes (QECCs) are necessary
to realize quantum communications and to make fault-tolerant quantum computers \cite{deutsch2020harnessing,preskill2018quantum}. The \emph{stabilizer} formalism provides a useful way to construct QECCs from classical   codes,  but certain orthogonal constraints  are required  \cite{calderbank1998quantum}. The entanglement-assisted (EA)  quantum  error  correction code (EAQECC)     \cite{brun2006correcting,PhysRevA.76.062313,brun2014catalytic} generalizes the stabilizer code. By presharing some entangled states between the sender (Alice) and the receiver (Bob), EAQECCs can be constructed from any classical linear codes without the orthogonal constraint. Therefore the construction  could be greatly simplified. As an important physical resource, entanglement can boost the classical information capacity of quantum channels \cite{PhysRevLett.83.3081,PhysRevLett.126.250501,holevo2002on,hsieh2008entanglement,hsieh2010trading,hsieh2010entanglement}. Recently, it has been shown that EAQECCs can  violate the nondegenerate quantum Hamming bound \cite{li2014entanglement} or the quantum Singleton bound \cite{grassl2021entanglement}.

Compared to standard QECCs, EAQECCs must establish some amount of entanglement before   transmission. This preshared entanglement is the price to be paid for enhanced communication capability. In a sense, we need to consider the \emph{net}  transmission of EAQECCs, i.e., the number of qubits transmitted minus that of ebits preshared. Further, it is   difficult to preserve too many noiseless ebits in EAQECCs at present. Thus we have to use as few ebits as possible to conduct the communication, e.g., one or two ebits are preferable \cite{hsieh2011high,hsieh2009entanglement,fujiwara2013a,wilde2013entanglement}. In addition, EAQECCs with positive net transmission and little entanglement  can lead to catalytic quantum codes \cite{brun2006correcting,brun2014catalytic}, which are applicable to fault-tolerant
 quantum computation (FTQC). In Ref.~\cite{brun2006correcting},  a table of best known EAQECCs of length up to $10$ was established through computer search or algebraic
methods.  
Several EAQECCs in Ref.~\cite{brun2006correcting} have larger minimum distances than the best known standard QECCs of the same length and net transmission. However, for   larger code lengths, the efficient construction of EAQECCs with better parameters than standard QECCs is still unknown.

In classical coding theory,    concatenated code (CC), originally proposed by Forney in 1960s \cite{forney1965concatenated}, provide a useful  way of constructing long codes from short ones.  CCs can  achieve very large coding gains   with  reasonable  encoding and decoding complexity \cite{lin2004error}. Moreover, CCs  can have  large minimum distances since  the distances of the component  codes are multiplied. As a result, CCs have been  widely used in many digital communication systems, e.g., the NASA standard for the Voyager program   \cite{costello2007channel}, and the compact disc (CD) \cite{lin2004error}.  Similarly in QECCs, the concatenated quantum codes  (CQCs), introduced by Knill and Laflamme in 1996 \cite{knill1996concatenated}, are  also effective  for constructing good
quantum codes. 
Particularly, it has been shown that CQCs are of great importance in realizing FTQC \cite{gottesman1997stabilizer,aharonov1997fault,campbell2017roads}.

Moreover, there exists a specific phenomenon in QECCs, called   \emph{error degeneracy},
 which distinguishes quantum codes from classical ones in essence. It is widely believed that degenerate codes can  correct more quantum errors than nondegenerate ones. Indeed, there are some  open problems concerning whether degenerate  codes can   violate the  nondegenerate quantum Hamming bound      \cite{sarvepalli2010degenerate} or can improve the quantum channel capacity \cite{divincenzo1998quantum,renes2015efficient}. Many CQCs   have been shown  to be degenerate even if the component codes are nondegenerate, e.g.,  {Shor's $[[9,1,3]]$} code and the $[[25,1,9]]$   CQC \cite{gottesman1997stabilizer,Grassl:codetables}. If we introduce extra entanglement to CQCs, it is possible to improve the error degeneracy performance of CQCs.

  In this article, we generalize the idea of concatenation to EAQECCs,  and  propose    entanglement-assisted concatenated quantum codes (EACQCs).
We show that EACQCs can beat the nondegenerate quantum Hamming bound while standard CQCs cannot.
Several families of degenerate EACQCs that can surpass the nondegenerate   Hamming bound for EAQECCs, are constructed.
 The same conclusion could be reached for the asymmetric error models,  in which the phase-flip errors ($Z$-errors) happen more frequently than the bit-flip errors ($X$-errors) \cite{sarvepalli2009asymmetric,fan2021asymmetric}.
Furthermore, we derive a number of EACQCs with better parameters
than the  best known QECCs and EAQECCs.
In particular, we see that many EACQCs have  positive net transmission  and  each of them consumes only one or two ebits.  Thus  they give rise to   catalytic EACQCs  with little entanglement and better parameters than the best known QECCs.   Further,  we show that the  EACQC scheme makes EA quantum communication possible even if the ebits are noisy. We compute the  entanglement fidelity (EF) of the $[[15,1,9;10]]$ EACQC by using Bowen's $[[3,1,3;2]]$ EAQECC \cite{bowen2002entanglement} or the $[[3,1,3;2]]$ EA repetition code \cite{brun2006correcting,brun2014catalytic} as the inner code. The outer code is the standard $[[5,1,3]]$ stabilizer code. We show that the $[[15,1,9;10]]$ EACQC   performs much better than the $[[25,1,9]]$ CQC over depolarizing channels if the ebits  suffering lower    error rate  than the qubits.  Moreover, we compute the  error probability threshold of EACQCs and we show
   that  EACQCs have much higher thresholds than CQCs   when  the error rate of ebits is sufficiently lower than that of qubits.

\section{Preliminaries}
\label{Preliminaries}
Let $q=2^m$ ($m\geq1$ is an integer) and denote by $GF(q)$  the extension  field of the binary field $GF(2)$.   Let $\mathbb{C}$ be the field of complex numbers, and let $V_n=(\mathbb{C}^q)^{\otimes n}=\mathbb{C}^{q^n}$ be  the $q^n$-dimensional Hilbert space, where $n$ is a positive integer.
Define two error operators on $\mathbb{C}^q$ by $X(a)|\psi\rangle=|a+\psi\rangle$   and $Z(b)|\psi\rangle=(-1)^{tr(
b\psi)}| \psi\rangle$,  where $a\in GF(q)$,  $b\in GF(q)$, and ``$tr$'' denotes  the trace operator from $GF(q)$ to $GF(2)$. For a vector $\mathbf{u}=(u_1,\cdots,u_n)\in GF(q)^n$, denote by $X(\mathbf{u})=X(u_1)\otimes\cdots\otimes X(u_n)$ and $Z(\mathbf{u})=Z(u_1)\otimes\cdots\otimes Z(u_n)$. Let $\Xi_n=\{X(\mathbf{a})Z(\mathbf{b})|\mathbf{a},\mathbf{b}\in GF(q)^n\}$
and let
$
\mathcal{G}_n=\{(-1)^u X(\mathbf{a})Z(\mathbf{b})|  \mathbf{a},\mathbf{b}\in GF(q)^n, u\in GF(2)\}
$ be the   group generated by $\Xi_n$. For the operator   $e=(-1)^uX(\mathbf{a})Z(\mathbf{b})\in \mathcal{G}_n$, the   weight of $e$ is defined by
$
\textrm{wt}_Q(e)=|\{1\leq i\leq n: (a_i,b_i)\neq(0,0)\}|.
$
The definition of quantum stabilizer codes is given below.
\begin{definition}
\label{definition of QEC and AQC}
A stabilizer code $ {Q}$ is a $q^k$-dimensional ($k\geq0$) subspace of $V_n$ such that
 $
 {Q}=\bigcap_{e\in T}\{|\phi\rangle\in V_n: e|\phi\rangle=|\phi\rangle\},
$
where $T$ is   a subgroup of $\mathcal{G}_n$.  $ {Q}=[[n,k,d]]_q$ has minimum distance $d$ if   it
can detect all  errors  $ e\in \mathcal{G}_n$ of  weight $\textrm{wt}_Q(e)$ up to $d-1$.
 Further, ${Q}$ is called nondegenerate if every    stabilizer in $T$  has   weight
    larger than or equal to   $d$, otherwise it is called degenerate.
\end{definition}

A concatenated quantum code (CQC) is derived from an inner code and an outer  code. In general,
 the component codes of CQCs can be chosen as stabilizer codes or non stabilizer codes. In this article, it suffices to   consider only
the case of stabilizer codes. Let the inner   and outer codes be $Q_I=[[n_1,k_1,d_1]]$ and $Q_O=[[n_2,k_2,d_2]]_{2^{k_1}}$, respectively.
Then we can derive a CQC   \cite{grassl2009generalized} with parameters
$
Q_C=[[ n_1n_2, k_2k_2,d_C\geq d_1d_2]].
$

An EAQECC with parameters $Q_e=[[n,k,d;c]]_q$ can encode  $k$ qudits into $n$ qudits by consuming   $c$ pairs
of maximally entangled states between Alice and Bob. It should be noted that  EAQECCs can be constructed from arbitrary classical linear codes directly.
The Calderbank-Shor-Steane (CSS) framework     \cite{steane1996error,calderbank1998quantum} provides a useful
way to construct both QECCs and EAQECCs from classical linear codes. 

\begin{lemma}[\cite{brun2006correcting}]
\label{CSSEAQECCs}
  Denote by $C_1=[n,k_1,d_1]_q$ and $C_2=[n,k_2,d_2]_q$   two linear codes over $GF(q)$. There exists an
  EAQECC   with parameters
$
Q_e=[[n,k_1+k_2-n+c,d_e\geq\min\{d_1,d_2\};c]]_q,
$
 where $c=rank(H_1H_2^T)$.
\end{lemma}

EAQECCs can also be constructed  by using the Hermitian construction \cite{calderbank1998quantum,brun2006correcting,wilde2008optimal} as follows.
\begin{lemma}[\cite{brun2006correcting}]
\label{HermitianEAQECCs}
Let $C=[n,k,d]_{q^2}$ be a linear code over $GF(q^2)$.  There exists an  EAQECC  with parameters
$
Q_e =[[n,2k -n+c,d_e \geq d ;c]]_q,
$
 where $c=rank(HH^\dagger)$, and $H^\dagger$ is the  conjugate transpose of $H$ over  $GF(q^2)$.
 \end{lemma}
 
We organize the main results of our study in the following order. Firstly, we present the construction of the entanglement-assisted concatenated quantum codes from two component quantum codes. Secondly, we construct several families of EACQCs violating the nondegenerate Hamming bound for EAQECCs. Thirdly, we derive a number of EACQCs with better parameters than the best known QECCs and EAQECCs. At last, we show that EACQCs can correct errors in the ebits. It is shown that EACQCs can outperform CQCs in entanglement fidelity and have higher error probability thresholds than CQCs.

\section{Entanglement-Assisted Concatenated Quantum Codes}\label{EAQECCs}
We generalize   CQCs to EACQCs by concatenating two quantum codes which can be chosen as either standard QECCs or EAQECCs. In this article,   sometimes we  represent an $ [[n,k,d]]_q$ QECC as an $[[n,k,d;0]]_q$ EAQECC so that we can unify the representation of QECCs and EAQECCs.
Let the inner   code  be $Q_I=[[n_1,k_1,d_1;c_1]]$, which requires $c_1$ ebits.
 Denote by ${k}^*_1\equiv k_1-c_1$ the net transmission of $Q_I$.
Let the outer code be $Q_{O}=[[n_2,k_2,d_2;c_2]]_{2^{k_1}}$, which can  either be binary or nonbinary  and
 that depends on ${k}_1$. $Q_O$ uses $c_2$ edits, or equivalently, $c_2{k}_1$ ebits. Denote by ${k}^*_2\equiv k_2-c_2$ the net transmission of $Q_O$.
  {Notice that, for classical linear codes  and quantum codes over the binary field $GF(2)$,  we usually neglect the   index  in the code parameters if there is no ambiguity.}

We have the following result about EACQCs.

\begin{theorem}\label{EACQCLemma}
Let  $Q_I=[[n_1,k_1,d_1;c_1]]$ be the inner code, and
let  $Q_{O}=[[n_2,k_2,d_2;c_2]]_{2^{{k}_1}}$ be the outer code.
 There exists an EACQC $\mathscr{Q}_e$ with parameters
 \begin{equation}
\mathscr{Q}_e=[[n_1n_2, k_1k_2,d_e\geq d_1d_2;c_e ]],
\end{equation}
where $c_e=c_1n_2+ c_2{k}_1$ is the number of ebits. The net transmission   is  $k_e^*= k_1k_2-c_e$.
\end{theorem}

 \begin{proof}
Based on the idea of code concatenation, we simply concatenate the inner code $Q_I$ with the outer code $Q_O$
to derive the EACQC \cite{forney1965concatenated,knill1996concatenated,grassl2009generalized}.
First we  encode the information state  $ |\mu\rangle$ by using the outer code $Q_O$, i.e.,
 \begin{equation}
 |\mu\rangle  \mapsto |\psi\rangle_O= (U_{O}\otimes \widehat{I}_{B_O})|\mu\rangle\otimes|0\rangle^{\otimes
(n_2-k_2-c_2) {k}_1}\otimes |\Psi_+\rangle_{AB}^{c_2{k}_1},
 \end{equation}
where there are  ${c_2{k}_1}$ Bell states, $|\Psi_+\rangle_{AB}^{c_2{k}_1}$, preshared between Alice and
Bob during the outer encoding. The outer encoding operation $U_O$ is applied to the qubits in Alice's side.
 Bob's halves of ebits are preshared and they do not need to be encoded.

Suppose that we can represent  $|\psi\rangle_O $ by
\begin{equation}
    |\psi\rangle_O =\sum_{\nu_1,\cdots,\nu_{n_2}=0}^{2^{k_1}}\ell_{\nu_1\cdots \nu_{n_2}}|\nu_1\cdots\nu_{n_2}\rangle,
\end{equation}
where $\ell_{\nu_1\cdots \nu_{n_2}} (0\leq\nu_1,\cdots,\nu_{n_2}\leq 2^{k_1} )$ should satisfy  the normalization condition. We separate each basis   $|\nu_1\cdots\nu_{n_2}\rangle$ in  $|\psi\rangle_O $ into $n_2$ subblocks, i.e., $|\nu_1\cdots\nu_{n_2}\rangle=|\nu_1\rangle\cdots|\nu_{n_2}\rangle$ for    $0\leq\nu_1,\cdots,\nu_{n_2}\leq 2^{k_1}$.
  For each subblock $| \nu_i\rangle (1\leq i\leq n_2)$,  we do the inner encoding as follows:
\begin{equation}
| \nu_i\rangle \mapsto |\psi_i\rangle_I= (U_{I}\otimes \widehat{I}_{B_I})| \nu_i\rangle\otimes|0\rangle^{\otimes n_1-k_1-c_1}\otimes |\Phi_+\rangle_{AB}^{c_1}.
 \end{equation}
  $|\Phi_+\rangle_{AB}^{c_1}$   are ${c_1 }$ Bell states preshared between Alice and
Bob during each inner   encoding. The inner encoding operation $U_I$ is applied to the qubits in Alice's side
while Bob's halves of ebits  do not need to be encoded.
It is easy to see that the    number of ebits used during the whole inner encoding is $c_1n_2$. The encoding process  of
EACQCs is given in Fig.~\ref{encodingEACQCs}.

The numbers of ebits  used during the outer and the inner encoding are equal to $c_2k_1$ and $c_1n_2$, respectively. Therefore the total number of ebits  is equal to $c_e=c_1n_2+ c_2{k}_1$. It is easy to see that the dimension of the EACQC  $\mathscr{Q}_e$  is equal to $2^{k_1k_2}$. Similar  to the principle of code concatenation in  \cite{forney1965concatenated,knill1996concatenated,grassl2009generalized}, the minimum distance of $\mathscr{Q}_e$ is at least $d_1d_2$. Therefore we can obtain an EAQECC with parameters
$
\mathscr{Q}_e=[[n_1n_2, k_1k_2,d_e\geq d_1d_2;c_e ]].
$
 \end{proof}
It is easy to see that if the inner and outer codes are both standard QECCs,  then the EACQC is a standard CQC.
Moreover, we can  use different inner codes in EACQCs.
 Let $Q_{I_i}=[[n_{I_i},k_{I_i},d_{I_i};c_{I_i}]]$ ($1\leq i\leq n_2$) be   $n_2$ inner codes.
   For simplicity,  we let $ {k}_1\equiv k_{I_1}=\ldots=k_{I_{n_2}}$, and
let $d_1\equiv d_{I_1}=\ldots=d_{I_{n_2}} $.  Let the outer code be $Q_{O}=[[n_2,k_2,d_2;c_2]]_{2^{{k}_1}}$. Then we can derive an EACQC  with parameters
 \begin{equation}
\mathscr{Q}_e'=[[\sum_{i=1}^{n_2}n_{I_i}, {k}_1k_2 ,d_e'\geq d_1d_2;c_e' ]] ,
\end{equation}
where $c_e'=\sum_{i=1}^{n_2}c_{I_i}+ c_2{k}_1$. The net transmission    is  $ k_1k_2-c_e'$.

 \begin{figure}
  \centering
  \includegraphics[width=3.4in]{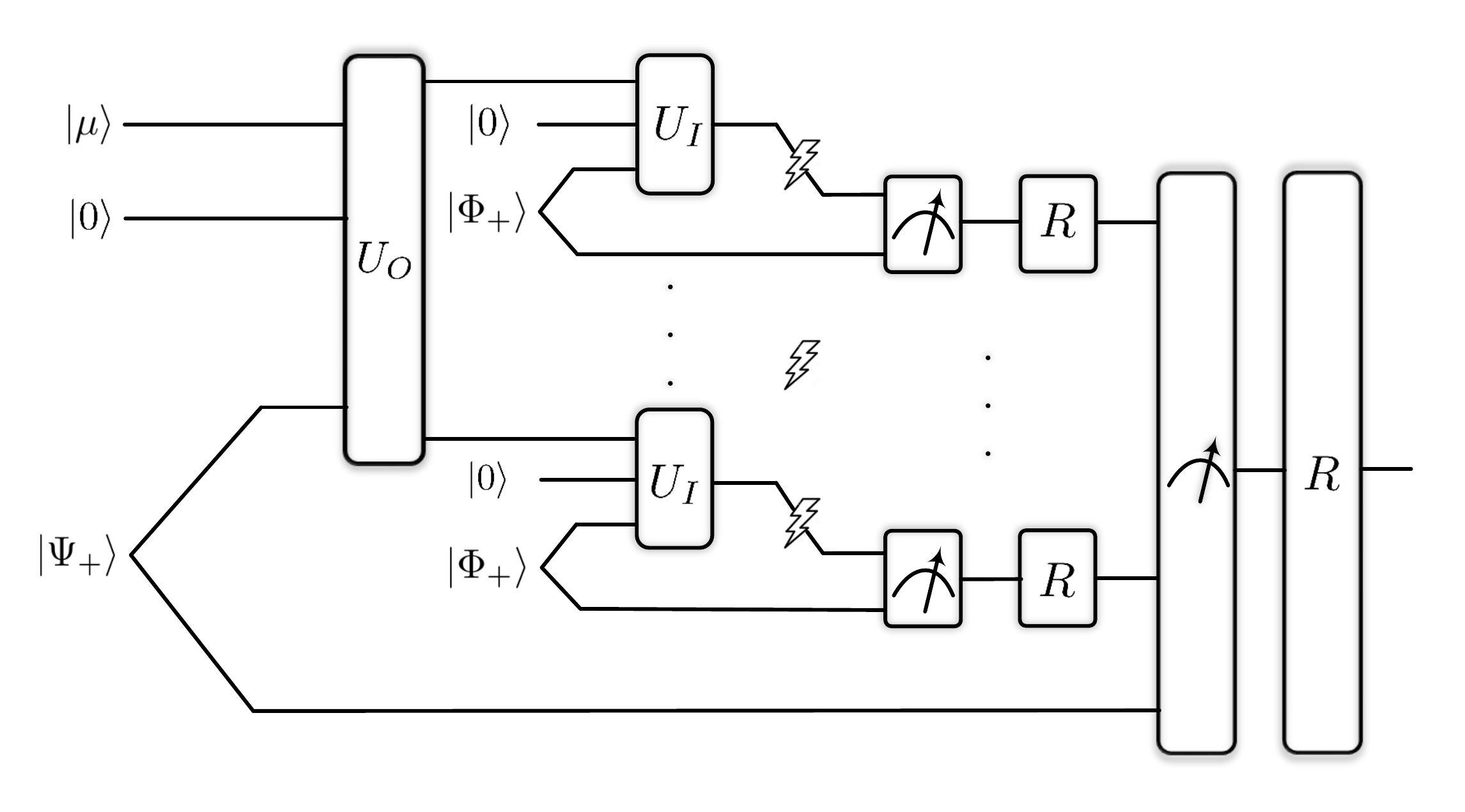}
  \caption{The encoding circuit of EACQCs. The information state $|\mu\rangle$ is first encoded with the outer encoder $U_O$ by presharing  $c_2{k}_1$  Bell states $|{\bf\Psi}_+\rangle =|\Psi_+\rangle_{AB}^{c_2{k}_1}$ between Alice and Bob. For    the output of $U_O$,  each subblock is encoded with the inner encoder $U_I$ by presharing  $c_1$  Bell states $|{\bf \Phi}_+\rangle =|\Phi_+\rangle_{AB}^{c_1}$ between Alice and Bob.}  \label{encodingEACQCs}
\end{figure}

\section{EACQCs Beating the Nondegenerate Quantum Hamming Bound}\label{ConstructionsofEACQCs}
Firstly, let us review the nondegenerate Hamming bound for EAQECCs \cite{lai2017linear}.
\begin{lemma}[\cite{lai2017linear}]
\label{EAQECCHammingBoundLemma}
For a binary nondegenerate  $Q_e=[[n,k,d;c]]$  EAQECC, it must satisfy
\begin{equation}
\label{EAQECCHammingBound}
 \sum_{i=0}^{\lfloor \frac{d-1}{2} \rfloor}3^i\binom{n}{i}\leq 2^{n+c-k}.
\end{equation}
Taking the limit as $n\rightarrow\infty$, this yields the asymptotic bound on the rate $k/n$:
\begin{equation}\label{asymptoticHamming}
\frac{k}{n}\leq 1+\frac{c}{n}-\frac{\delta}{2}\log_23-H_2\left(\frac{\delta}{2}\right),
\end{equation}
where $\delta=d/n$, and $H_2(x)=-x\log_2x-(1-x)\log_2(1-x)$ is the binary entropy function.
\end{lemma}

To the best of our knowledge, there do not exist CQCs that violate the quantum Hamming bounds  \cite{lai2017linear}. However, the situation changes significantly in the entanglement-assisted scenarios. We could easily construct several families of EACQCs that  violate the nondegenerate Hamming bound in Lemma~\ref{EAQECCHammingBoundLemma}. We summarize these EACQCs as follows.
\begin{theorem} There exist the following four families of EACQCs with parameters
\begin{itemize}
\item  [$(\rm{I})$] $\mathscr{Q}_{e_1} = [[5n_2,1,d_{e_1}\geq 3n_2;n_2-1 ]],$
where $n_2\geq3$ is odd.
\item  [$(\rm{II})$] $\widetilde{\mathscr{Q}}_{e_1} = [[5n_2,1,\widetilde{d}_{e_1}\geq 3n_2-3;n_2-1 ]],$
where $n_2\geq10 $ is even.
\item  [$(\rm{III})$] $\mathscr{Q}_{e_2} = [[4n_2,1,d_{e_2}\geq 3n_2;2n_2-1 ]],$
where $n_2\geq11$ is odd.
\item  [$(\rm{IV})$]$\widetilde{\mathscr{Q}}_{e_2} = [[4n_2,1,\widetilde{d}_{e_2}\geq 3n_2-3;2n_2-1 ]],$
where  $n_2\geq 32$ is even.
\end{itemize}
   EACQCs in $(\rm{I})-(\rm{IV})$  can beat the nondegenerate quantum Hamming bound for EAQECCs, respectively.
\end{theorem}
\begin{proof}
The proof is given in Appendix \ref{App1}.
\end{proof}

We give an explicit example to illustrate the construction of EACQCs. Let $Q_I=[[5,1,3;0]]$ be the inner code  and let $Q_O=[[3,1,3;2]$ be the outer code in Ref.~\cite{brun2006correcting}. Then we can derive an EACQC with parameters $\mathscr{Q}_e=[[15,1,9;2]]$ by Theorem \ref{EACQCLemma}. This code can beat the nondegenerate Hamming bound for EAQECCs in Eq.~(\ref{EAQECCHammingBound}). Notice that $Q_I=[[5,1,3;0]]$ and $Q_O=[[3,1,3;2]$ are
both nondegenerate codes \cite{calderbank1998quantum,brun2006correcting}, while     $\mathscr{Q}_e=[[15,1,9;2]]$ is degenerate.  Also notice that $Q_I$ and $Q_O$ cannot beat the nondegenerate Hamming bound in Eq.~(\ref{EAQECCHammingBound}), but their EACQC $\mathscr{Q}_e=[[15,1,9;2]]$ can do so. If we encode one of  the   qubits of the outer encoding by using the $[[4,1,3;1]]$ EAQECC, then we   derive a $[[14,1,9;3]]$ EACQC.  This code can also beat the nondegenerate Hamming bound for EAQECCs.

For the asymmetric channel models, we present the construction of     EACQCs that can beat the nondegenerate Hamming bound for asymmetric EAQECCs.  Let $d_X$ and $d_Z$ be two positive integers. From \cite{galindo2020asymmetric}, an asymmetric EAQECC  $Q_A=[[n,k,d_Z/d_X;c]]_q$ can
detect any $X$-error    of   weight  up to  $d_X-1$   and any $Z$-error of weight up to  $d_Z-1$, simultaneously. The number of edits is $c$.  One could further obtain nondegenerate Hamming bounds for asymmetric EAQECCs \cite{galindo2020asymmetric,lai2017linear}.
\begin{lemma}[\cite{galindo2020asymmetric}]\label{AsymmetricEAQECCHammingBound}
A binary nondegenerate asymmetric EAQECC $[[n,k,d_Z/d_X;c]]$ must satisfy
\begin{equation}
 \sum_{i=0}^{\lfloor \frac{d_X-1}{2}\rfloor} \binom{n}{i} \sum_{j=0}^{\lfloor \frac{d_Z-1}{2}\rfloor} \binom{n}{j}\leq 2^{n+c-k}.
\end{equation}
\end{lemma}

Let $Q_I=[[n_1,1,n_1/1;0]]$ be a binary asymmetric EAQECC which is used as the inner code. Let $Q_O=[[n_2,1,d_2/d_2;n_2-1]]$ be the outer code, where $d_2=n_2-1$ for even $n_2\geq 2$, or $d_2=n_2$ for odd $n_2\geq3$. We concatenate $Q_I$ with $Q_O$ according to Fig. \ref{encodingEACQCs}.   Then we    have the following result about asymmetric EACQCs.
\begin{corollary}\label{asymmetricEACQCs1}
There exists a family of asymmetric EACQCs with parameters
\begin{equation}
\mathscr{Q}_A=[[n_1n_2,1,n_1d_2/d_2;n_2-1]],
\end{equation}
where $n_1\geq 2$ is an integer, $d_2=n_2-1$ for even $n_2\geq2$, or $d_2=n_2 $ for odd $n_2\geq3$.
\end{corollary}

For any integer $n_1\geq2$ and any odd $n_2\geq 3$, $\mathscr{Q}_A$ in Corollary \ref{asymmetricEACQCs1} can beat the nondegenerate Hamming bound for asymmetric EAQECCs in Lemma \ref{AsymmetricEAQECCHammingBound}. For any integer $n_1\geq 2$ and any even $n_2 \geq 8$, $\mathscr{Q}_A$ can also beat the nondegenerate Hamming bound. Let $n_1=2$ and $n_2=3$. We can derive an asymmetric EACQC with parameters $\mathscr{Q}_A=[[6,1,6/3;2]]$. It is the  shortest length   EAQECC that can beat the nondegenerate Hamming bound known to date.

\section{EACQCs beating existing QECCs and EAQECCs}\label{ConstructionsofEACQCs}

  Similar to classical coding theory, constructing quantum codes with parameters better than the best known results is one central topic in quantum coding theory. It is even more attractive  since degenerate quantum codes have large potentials to outperform  any nondegenerate quantum code. Indeed, a number of best known QECCs in \cite{Grassl:codetables} have been shown to be degenerate. 

As argued in Ref.~\cite{brun2006correcting}, we say that
an EAQECC $[[n,k_1,d;c]]$ is better than a QECC $[[n,k,d]]$ if the net transmission  $(k_1-c)$ is larger than $k$. Ref.~\cite{Grassl:codetables} collects a list of classical linear codes and QECCs with best parameters currently known. According to the construction of EAQECCs in Lemma \ref{HermitianEAQECCs},  a  quaternary code in  Ref.~\cite{Grassl:codetables}    corresponds to a best known nondegenearate EAQECCs. In general, it is not difficult to construct nondegenearate EAQECCs with postive net transmissions better than the best known QECCs based on Ref.~\cite{Grassl:codetables}. Therefore we should focus on constructing  (degenerate) EAQECCs with positive net transmissions that can beat  the best known nondegenerate EAQECCs. This is an important topic   concerning that whether degeneracy can help to improve the classical coding limit in EAQECCs.

 We give two explicit  constructions to show that   EACQCs can beat the best known QECCs and     EAQECCs.  According to \cite{macwilliams1981theory,seroussi1986on},    there exists a cyclic maximum-distance-separable (MDS) code with parameters $[17,9,9]_{16} $. From Lemma \ref{HermitianEAQECCs}, we can derive an entanglement-assisted quantum maximum-distance-separable (EAQMDS) code with parameters $[[17,5,9;4]]_4$. Let $Q_I=[[4,2,2]]$ be the inner code and let $Q_O=[[17,5,9;4]]_4$ be the outer code. Then we can derive an EACQC with parameters $\mathscr{Q}_e=[[68,10,18;8]]$.  Compared with the best known  $Q=[[68,2,16]]$ QECC in   \cite{Grassl:codetables},   the EACQC $\mathscr{Q}_e$ has  a larger minimum distance
while maintaining the same length and net transmission. $\mathscr{Q}_e$ also has a larger minimum distance than the best known nondegenerate $ [[68,10,16;8]]$ EAQECC  from \cite{Grassl:codetables} of the same length and net transmission.
Let
$Q_O=[[65,17,33;16]]_8$ be an EAQMDS code constructed from a cyclic MDS code $[65,33,33]$ in \cite{macwilliams1981theory} and let $Q_I=[[8,3,3;0]]$. Then we can derive an EACQC with parameters $\mathscr{Q}_e=[[520,51,99;48]]$ by using $Q_O=[[65,17,33;16]]_8$ and $Q_I=[[8,3,3;0]]$ as the outer and inner codes, respectively. This EACQC is better than the asymptotic Gilbert-Varshamov (GV) bound for EAQECCs in Ref.~\cite{lai2017linear}.   In  Appendix \ref{App2}, Table \ref{tables1} and Table \ref{tables2}, we list more constructions of EACQCs with parameters better than the best known QECCs and EAQECCs.

 In practice, we prefer to use as few as possible ebits to do the entanglement-assisted
communication since  storing  a large number of
noiseless ebits is quite difficult.
Let $Q_I=[[5,1,3;0]]$ be the inner code and let $Q_O=[[3,2,2;1]]$  be the outer code,
 then we can derive a  $[[15,2,6;1]]$ EACQC.
 This code has   larger minimum distance than the best known standard $[[15,1,5]]$ QECC   in \cite{Grassl:codetables}.
 By using the MAGMA software \cite{cannon2006handbook}, we know that there exists  a nondegenerate $[[15, 8,6;7]]$ EAQECC. This code has the same minimum distance and net transmission with the $[[15,2,6;1]]$ EACQC. However, the EACQC consumes only one ebit and thus it is more practical. In  Appendix \ref{App2}, Table  \ref{tables3}, we list a number of EACQCs with   parameters better than the best known QECCs and EAQECCs, and each EACQC consumes only one ebit.

 In \cite{fan2016constructions}, sevral families of       $q$-ary  EAQMDS  codes with distances larger than $q+1$ and consuming
  very few   edits were constructed. We use EAQMDS codes in  \cite{fan2016constructions} as the outer code to construct EACQCs that consume very few ebits. We give an example to illustrate the construction.  Let $Q_I=[[4,2,2;0]]$ be the inner code  and let a $Q_O=[[17,4,8;1]]_{4}$   EAQMDS code in  \cite{fan2016constructions} be the outer code. Then we can derive an EACQC with parameters $\mathscr{Q}_e=[[68,8,16;2]]$.
  This code has a  larger minimum distance  than the best known    $[[68,6,14]]$  QECC  in    \cite{Grassl:codetables} of the same length and net transmission. It also has a larger minimum distance than the best known  nondegenerate $[[68,19,15;13]]$ EAQECC in   \cite{Grassl:codetables} of the same length and net transmission. In   Appendix \ref{App2}, Table  \ref{tables4}, we list a number of EACQCs   with better parameters than    the best known      QECCs and EAQECCs  in \cite{Grassl:codetables}, and
  each code  consumes  only a few
  ebits.

  \section{Performance of EACQCs with Noisy Ebits}
\label{ChannelFidelity}
 In this section, we evaluate the performance of EACQCs with noisy  ebits.  We compute the entanglement fidelity  and the error probability threshold of EACQCs and make comparisons  with the standard CQCs. For a quantum channel, the use of a QECC should improve the entanglement fidelity  when the error probability  is below  a specific   value,
 that we call it  \emph{the threshold}. In practical applications,  QECCs with sufficiently high thresholds are needed. We will show that EACQCs can outperform CQCs in entanglement fidelity if the ebits are less noisy than the qubits. Further, we will show that the threshold of EACQCs is much higher  than that of   CQCs  when the error probability of ebits is sufficiently  lower than that of
qubits.

During the process of entanglement-assisted quantum communication, the preshared  ebits of Bob need to be stored faultlessly and EAQECCs can only correct errors on the transmitted qubits. However, noise in Bob's ebits may be inevitable in practical applications \cite{lai2012entanglement,brun2014catalytic}, and maintaining a large number of noiseless  ebits is extremely difficult.
In this article, we use EACQCs to correct   errors in ebits.
 In the EACQC scheme, suppose that we use an EAQECC $Q_I$ as the inner code  and use a standard stabilizer code $Q_O$  as the outer code. We show that the outer code  $Q_O$  can not only correct errors on the physical qubits but also can correct errors on the ebits. We construct two EACQCs and show that   they can outperform CQCs in entanglement fidelity when the error probability of ebits is lower than that of qubits. We construct a $\mathscr{Q}_B=[[15,1,9;10]]^B$ EACQC by using the  $ [[5,1,3]]$ stabilizer code as the outer code and Bowen's $ [[3,1,3;2]]$  EAQECC \cite{bowen2002entanglement} as the inner code.  Alternately, we    use  the $[[5,1,3]]$ stabilizer code as the outer code, and the $[[3,1,3;2]]$ EA repetition code as the inner code to construct another  $\mathscr{Q}_R=[[15,1,9;10]]^R$ EACQC  with the same parameters. Recall that the standard $Q_C=[[25,1,9]]$ CQC  is the concatenation of the $[[5,1,3]]$ stabilizer code.
It is known that Bowen's $ [[3,1,3;2]]$ EAQECC is equivalent to the $ [[5,1,3]]$ stabilizer code and they have the same stabilizers. Thus the $\mathscr{Q}_B=[[15,1,9;10]]^B$ EACQC is equivalent to the $Q_C=[[25,1,9]]$ CQC. Then the $\mathscr{Q}_B=[[15,1,9;10]]^B$ EACQC has the same error correction ability with the $Q_C=[[25,1,9]]$ CQC. Nevertheless,
we show that EACQCs can outperform CQCs in entanglement fidelity if the error probability of ebits is lower than that of qubits.

\begin{figure*}[ht]
\hspace{-10pt}
\begin{minipage}{0.46\linewidth}
 \centering
  { \includegraphics[width=3.0in]{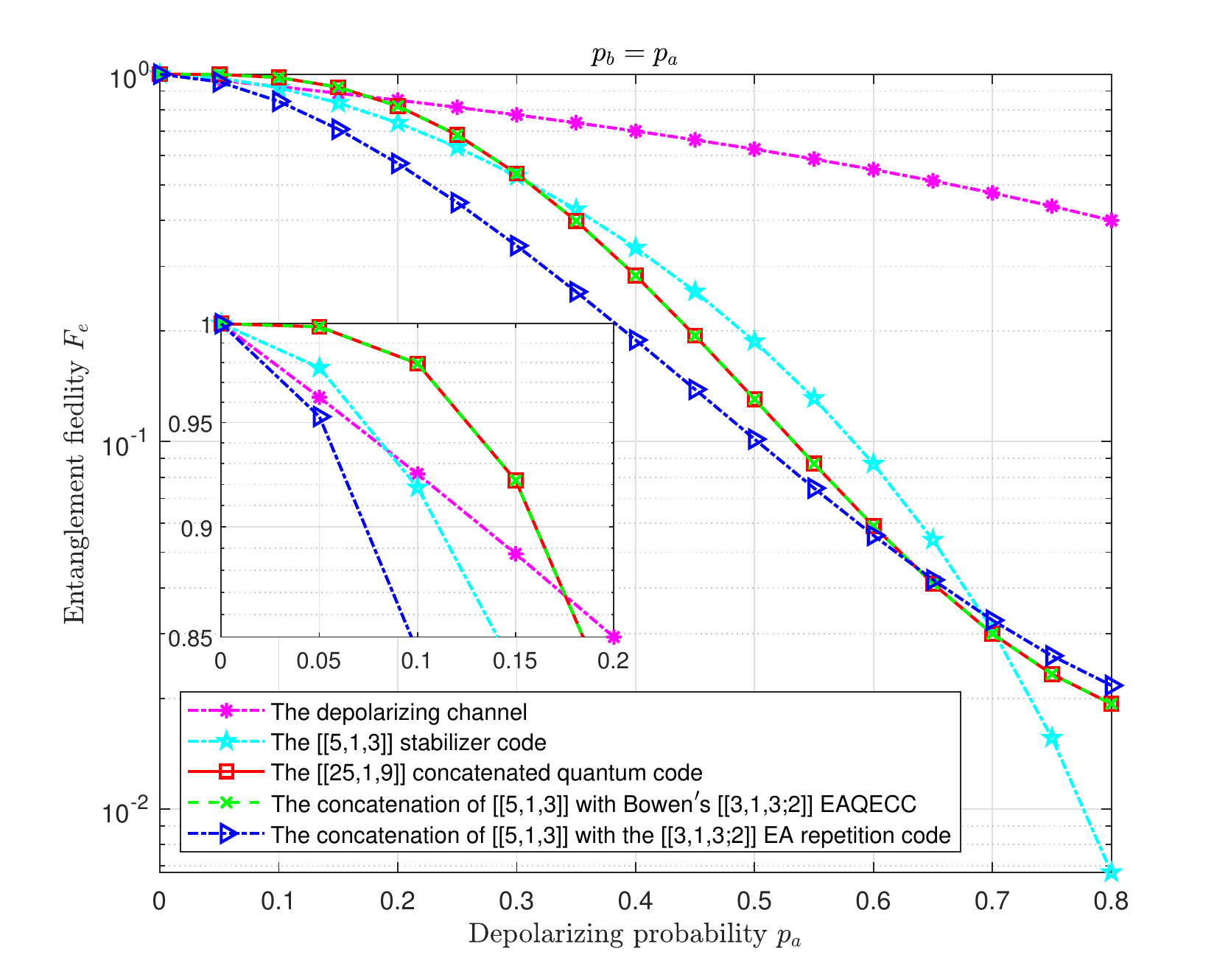}}
\end{minipage}
\begin{minipage}{0.46\linewidth}
\centering
  { \includegraphics[width=3.0in]{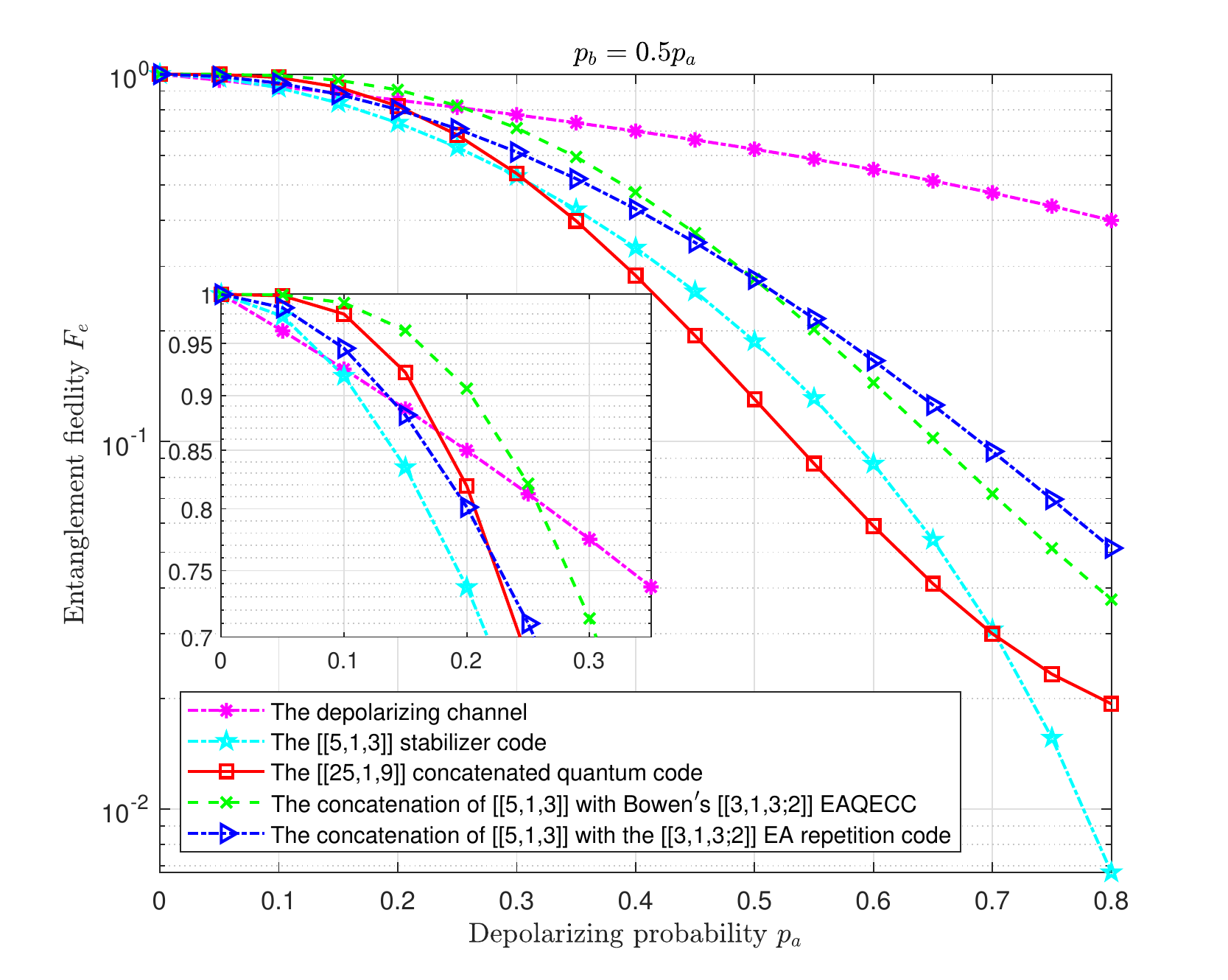}}
\end{minipage}

\vspace{3pt}
\hspace{-10pt}
\begin{minipage}{0.46\linewidth}
 \centering
  { \includegraphics[width=3.0in]{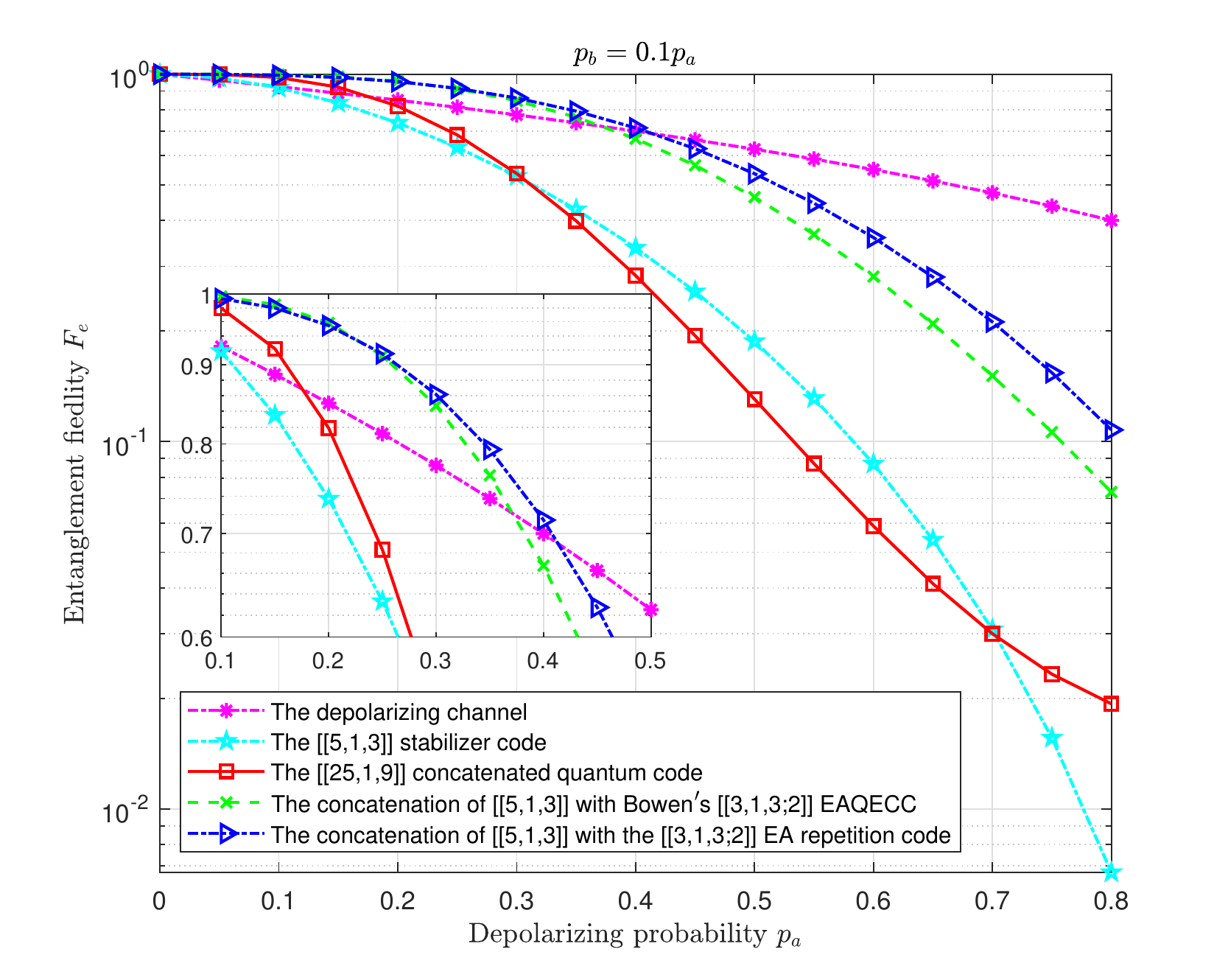}}
\end{minipage}
\begin{minipage}{0.46\linewidth}
 \centering
  { \includegraphics[width=3.0in]{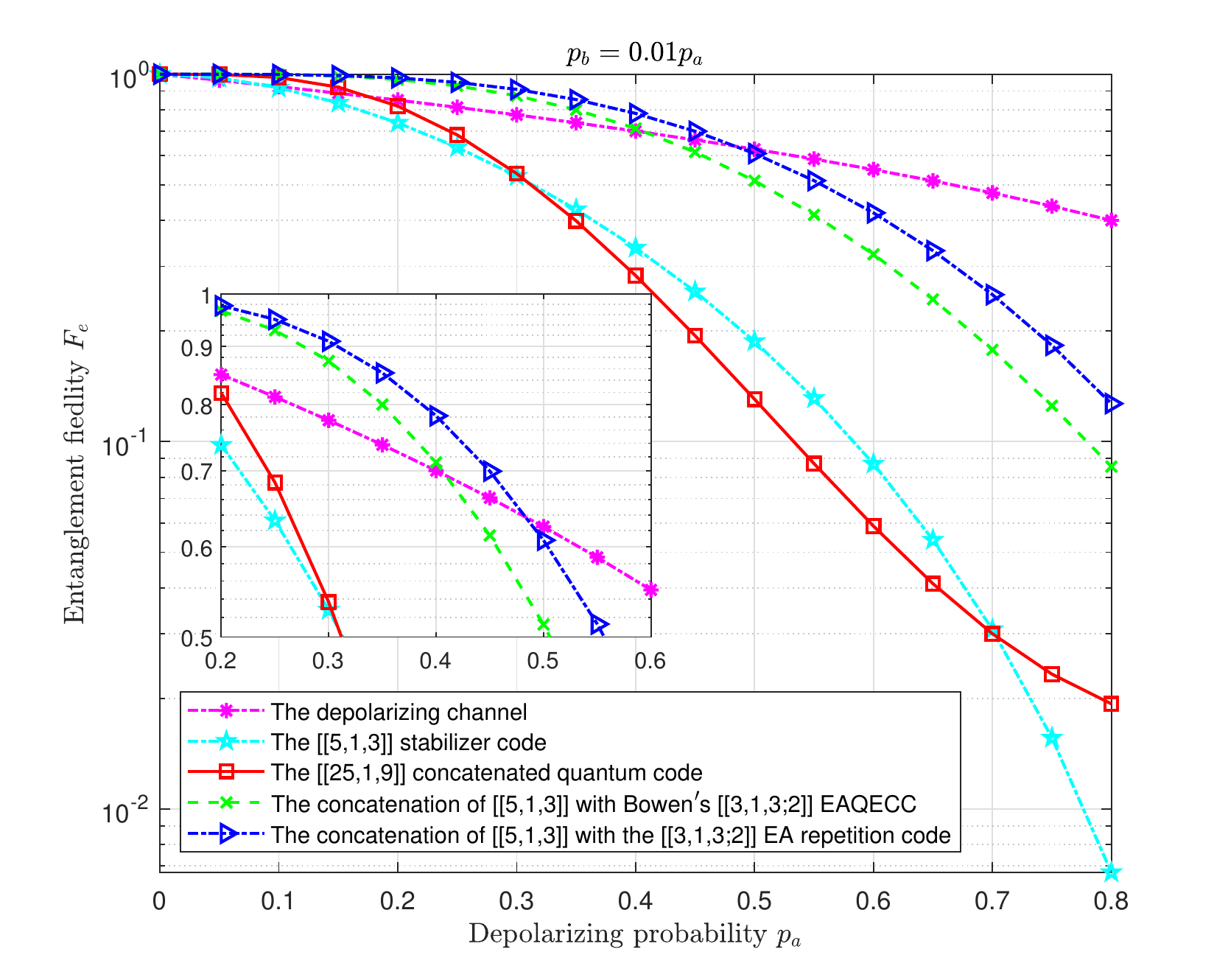}}
\end{minipage}
  \caption{Entanglement fidelity of EACQCs and CQCs for $p_b=p_a,0.5p_a,0.1p_a,0.01p_a$.   }\label{FC1}
\end{figure*}

The detailed entanglement fideltiy computation of the   two EACQCs and the CQC  was put  in   Appendix \ref{App3}.   The EFs of the two EACQCs and the CQC   were plotted   in Fig. \ref{FC1}.  We compare the EF of EACQCs with that of the $[[25,1,9]]$ CQC. If $p_a=p_b$, the EF of the  $\mathscr{Q}_B=[[15,1,9;10]]^B$ EACQC  is equal to that of the $[[25,1,9]]$ CQC. When $p_b=0.5p_a$, the EF of the
$\mathscr{Q}_B=[[15,1,9;10]]^B$ and the $\mathscr{Q}_R=[[15,1,9;10]]^R$ EACQCs can outperform      that of the $[[25,1,9]]$ CQC. As $p_b$ becomes even lower than $p_a$, e.g., $p_b=0.1p_a,0.01p_a$, the EF  of   $\mathscr{Q}_B=[[15,1,9;10]]^B$ and   $\mathscr{Q}_R=[[15,1,9;10]]^R$   performs  much better than that of the $[[25,1,9]]$ CQC. Moreover,     $\mathscr{Q}_B=[[15,1,9;10]]^B$   performs better than    $\mathscr{Q}_R=[[15,1,9;10]]^R$    when $p_b=p_a$.     While $p_b=0.1p_a,0.01p_a$,   $\mathscr{Q}_R=[[15,1,9;10]]^R$   performs much better than
  $\mathscr{Q}_B=[[15,1,9;10]]^B$   and the $[[25,1,9]]$ CQC.

We compare the error probability threshold  of the two EACQCs with that of the  CQC. For the $[[5,1,3]]$ stabilizer code and the $[[25,1,9]]$ CQC, the thresholds are $p>0.09$  and $p>0.18$, respectively.  Thus the CQC scheme can improve the error probability threshold.
For the EACQCs, when $p_b=0.5p_a$, the thresholds of   $\mathscr{Q}_B=[[15,1,9;10]]^B$   and      $\mathscr{Q}_R=[[15,1,9;10]]^R$   are $p>0.25$ and $p>0.14$, respectively. While  $p_b $ becomes sufficiently lower, e.g.,  $p_b=0.01p_a$, the thresholds  of   $\mathscr{Q}_B=[[15,1,9;10]]^B$  and  $\mathscr{Q}_R=[[15,1,9;10]]^R$   are $p>0.41$ and $p>0.47$, respectively. Therefore the EACQC scheme can greatly improve the error probability threshold   when the error probability of ebits is much lower than that of qubits.


 \section{Conclusions and Discussions}
\label{Conclusions}
In this article, we have proposed the construction  of entanglement-assisted concatenated quantum codes by concatenating an inner code with an outer code. We  not only have generalized the idea of concatenation to EAQECCs but also have shown that EACQCs can outperform many existed results.
We have further shown that EACQCs can beat the nondegenerate Hamming bound for EAQECCs while the standard CQCs cannot do so.  We have derived many EACQCs with   larger minimum distances than the best known  QECCs and EAQECCs in \cite{Grassl:codetables} of the same length and net transmission.
In addition, we have constructed several catalytic EACQCs with little entanglement and better parameters than the best known QECCs and EAQECCs. We have also constructed a family of asymmetric EACQCs that can beat the nondegenerate Hamming bound for asymmetric EAQECCs.  Finally, we have computed the entanglement fidelity of two EACQCs and compared them with the $[[25,1,9]]$  CQC. We have shown that EACQCs can outperform CQCs in entanglement fidelity when the ebits are less noisy than qubits. In particular, we have shown that EACQCs   have  much higher error thresholds than CQCs  when the error probability of ebits is sufficiently lower than that of qubits.   These properties of EACQCs make them     very competitive with standard CQCs for both   quantum communication and fault-tolerant quantum computation.

\bibliography{eacqc}

\appendices

\section{Proof of Theorem 2}
\label{App1}
We need the following estimation of a sum of binomial coefficient.

\begin{lemma}[\cite{macwilliams1981theory}]
Let $n\geq 2 $ be an integer. For $1<i<n$, there is
\begin{equation}
\frac{1}{\sqrt{8n\alpha(1-\alpha)}} 2^{nH_2(\alpha)}\leq \binom{n}{i} \leq   \frac{1}{\sqrt{2\pi n\alpha(1-\alpha)}} 2^{nH_2(\alpha)},
\end{equation}
where $\alpha=i/n$, and $H_2(x)=-x\log_2x-(1-x)\log_2(1-x)$ is the binary entropy function.
\end{lemma}

\setcounter{theorem}{1}
\begin{theorem}
 There exist the following four families of EACQCs with parameters
\begin{itemize}
\item  [$(\rm{I})$] $\mathscr{Q}_{e_1} = [[5n_2,1,d_{e_1}\geq 3n_2;n_2-1 ]],$
where $n_2\geq3$ is odd.
\item  [$(\rm{II})$] $\widetilde{\mathscr{Q}}_{e_1} = [[5n_2,1,\widetilde{d}_{e_1}\geq 3n_2-3;n_2-1 ]],$
where $n_2\geq10 $ is even.
\item  [$(\rm{III})$] $\mathscr{Q}_{e_2} = [[4n_2,1,d_{e_2}\geq 3n_2;2n_2-1 ]],$
where $n_2\geq11$ is odd.
\item  [$(\rm{IV})$]$\widetilde{\mathscr{Q}}_{e_2} = [[4n_2,1,\widetilde{d}_{e_2}\geq 3n_2-3;2n_2-1 ]],$
where  $n_2\geq 32$ is even.
\end{itemize}
   EACQCs in $(\rm{I})-(\rm{IV})$  can beat the nondegenerate quantum Hamming bound for EAQECCs, respectively.
\end{theorem}

 \begin{proof}
 We now elaborate the constructions of these EACQCs and their properties.  The EACQCs in $(\rm{I})-(\rm{IV})$ were all
 constructed according to  Theorem 1.

The EACQC $\mathscr{Q}_{e_1} $ in $(\rm{I})$ was constructed   by choosing the inner code $Q_I=[[5,1,3;0]]$ and the outer code  $Q_{O }=[[n_2,1,n_2;n_2-1]]$ in Ref.~\cite{lai2013duality}, where $n_2\geq3$ is odd.
The EACQC $\widetilde{\mathscr{Q}}_{e_1}  $ $(\rm{II})$ was constructed similarly by choosing the inner code $Q_I=[[5,1,3;0]]$ and the outer code  $Q_O=[[n_2,1,n_2-1;n_2-1]]$ in Ref.~\cite{lai2013duality}, where $n_2\geq 4 $ is even.
While letting the inner and outer codes   be $Q_I=[[4,1,3;1]]$ and    $Q_{O }=[[n_2,1,n_2;n_2-1]]$, respectively, where $n_2\geq3$   is odd,  the EACQC    $\mathscr{Q}_{e_2} $ in $(\rm{III})$ was derived.
The EACQC $\widetilde{\mathscr{Q}}_{e_2} $ in $(\rm{IV})$ was derived by
choosing  the inner and outer codes   as $Q_I=[[4,1,3;1]]$ and    $Q_{O }=[[n_2,1,n_2-1;n_2-1]]$, respectively, where $n_2\geq4$ is even.

 We need to prove that EACQCs in $\rm{(I)-(IV)}$ can beat the nondegenerate Hamming bound for EAQECCs.  For the $\mathscr{Q}_{e_1} = [[5n_2,1,d_{e_1}\geq 3n_2;n_2-1 ]] $ EACQC   in  $\rm{(I)}$, let $n_2=2m_2+1(m_2\geq1) $. According to Lemma \textbf{S1}, we have
 \begin{equation}
 \footnotesize
     \sum_{i=0}^{\lfloor\frac{3n_2-1}{2}\rfloor}3^i\binom{5n_2}{i}\geq 3^{3m_2+1}\binom{10m_2+5}{3m_2+1}\geq
     \frac{3^{3m_2+1}}{\sqrt{8(10m_2+5)\alpha_1(1-\alpha_1)}} 2^{(10m_2+5)H_2(\alpha_1)}\geq 2^{12m_2+4}=2^{6n_2-2}
 \end{equation}
 for all $m_2\geq1$, where $\alpha_1=(3m_2+1)/(10m_2+1)$. Thus the   EACQC  $\mathscr{Q}_{e_1}$ in  $\rm{(I)}$ can beat the nondegenerate Hamming bound for EAQECCs. For EACQCs in $\rm{(II)-(IV)}$, we can get similar results based on Lemma \textbf{S1}.

  In addition,  it is easy to verify that the relative distance of each EACQC  in $(\rm{I})-(\rm{IV})$
   violates  the asymptotic bound of nondegenerate Hamming bound for EAQECCs as the code length goes to infinity.
 \end{proof}
 \section{EACQCs Beating the Best Known  QECCs and EAQECCs}
 \label{App2}
In \cite{fan2016constructions,qian2019constructions}, entanglement-assisted
quantum maximum-distance-separable (EAQMDS) codes achieving the Singleton bound were constructed.
 We have the following general result about the construction of EAQMDS codes.
\begin{lemma}\label{EAQMDSGeneral}
For any $q+1\leq n\leq q^2+1$ and any integer $2\leq d< n/2-1 $, there exist    EAQMDS codes $
Q=[[n,k,d;c]]_q
$ with positive net transmission, i.e., $k>c$,  where $n+c=k+2d-2$,  and $0\leq c\leq n$.
\end{lemma}
\begin{proof}
It is known that there exist $q^2$-ary classical MDS codes $C=[n,k_1,d]_{q^2}$ of length $q+1\leq n\leq q^2+1$ \cite{seroussi1986on}, where $d=n-k_1 +1$. Let the parity-check matrix of $C$ be $H$ and denote by $c=rank(HH^\dag)$ the rank of $HH^\dag$. According to the Hermitian construction in \cite{brun2006correcting,wilde2008optimal}, there exist EAQMDS codes with parameters $
Q_e=[[n,2k_1-n+c,d;c]]_q
$. If $2\leq d< n/2-1 $, then we have $k=2k_1-n_1+c>c$.
\end{proof}

  In Appendix \ref{App2}, Table \ref{tables1} and Table \ref{tables2}, we list a number of EACQCs constructed  by using EAQMDS codes in Lemma \ref{EAQMDSGeneral} as the outer codes.

\begin{table*}
\tiny
 \caption{\scriptsize EACQCs with better paramters than   the best known QECCs  and EAQECCs of the same length and net transmission. {  The notation of the inner codes, for example, ``{$15\times[[4,2,2;0]]$+$2\times[[5,2,2;0]]$}'',  means
 that the inner codes are a mix of fifteen  $[[4,2,2;0]]$ codes and two $[[5,2,2;0]]$ codes.  The  outer codes are  EAQMDS codes in Lemma S2. The parameters $ k^*_2$, $k^*_e$ and $K^*_e$ stand for  the net transmissions of the $[[n_2,k^*_2,d_2]]_{2^{k_1}}$ outer code, the $ [[n_e,k^*_e,d_e ]]$ EACQC  and the $ [[N_e,K^*_e,D_e ]]$ EAQECC, respectively.   EAQECCs in the last column are obtained from the best known quaternary  codes in Ref.~\cite{Grassl:codetables}. For the QECCs and the EAQECCs  {missing explicit constructions} in  Ref.~\cite{Grassl:codetables}, their  parameters are in \textbf{bold} type. }}
\begin{center}
\begin{tabular}{lllll}
 \makecell[l]{$[[n_1,k_1,d_1;c_1]]$ \\Inner Codes}     &  \makecell[l]{$[[n_2,k^*_2,d_2]]_{2^{k_1}}$ \\Outer Codes}   &\makecell[l]{$[[n_e,k^*_e,d_e ]]$\\EACQCs }&    \makecell[l]{$[[n ,k ,d ]] $ \\ QECCs in  Ref.~\cite{Grassl:codetables}}&\makecell[l]{$[[N_e,K^*_e,D_e ]]$\\EAQECCs  from  Ref.~\cite{Grassl:codetables}}\\
\hline
$17\times[[4,2,2;0]]$  & $[[17, 1,9 ]]_4$ &$ [[68, 2, d_e\geq18]]$& $ [[68, 2, 16]]$  &$ [[68, 2, 16]]$ \\
\makecell[l]{$16\times[[4,2,2;0]]$+$[[5,2,2;0]]$} &  $[[17,1 ,9 ]]_4$ & $ [[69, 2 , d_e\geq18]]$&$ [[69, 2,  16]]  $ &$ [[69, 1 , 17]]$\\
\makecell[l]{$15\times[[4,2,2;0]]$+$2\times[[5,2,2;0]]$} &  $[[17,1 ,9 ]]_4$ & $ [[70, 2 , d_e\geq18]]$&$ [[70, 2,  16]]  $ &$ [[70, 2, 17]]$\\
\makecell[l]{$14\times[[4,2,2;0]]$+$3\times[[5,2,2;0]]$} &  $[[17,1 ,9 ]]_4$ & $ [[71, 2 , d_e\geq18]]$&$ [[71, 2,  16]]  $ &$ [[71, 1, 18]]$\\
\makecell[l]{$12\times[[4,2,2;0]]$+$5\times[[5,2,2;0]]$} &  $[[17,1 ,9 ]]_4$ & $ [[73, 2 , d_e\geq18]]$&$ [[73, 2,  16]]  $ &$ [[73, 1, 18]]$\\
\makecell[l]{$11\times[[4,2,2;0]]$+$6\times[[5,2,2;0]]$} &  $[[17,1 ,9 ]]_4$ & $ [[74, 2 , d_e\geq18]]$&$ [[74, 2,  16]]  $ &$ [[74, 2, 17]]$\\
\makecell[l]{$10\times[[4,2,2;0]]$+$7\times[[5,2,2;0]]$} &  $[[17,1 ,9 ]]_4$ & $ [[75, 2 , d_e\geq18]]$&$ [[75, 2,  17]]  $ &$ [[75, 1, 18]]$\\
\makecell[l]{$16\times[[4,2,2;0]]$+$[[3,2,2;1]]$} &  $[[17,1 ,9 ]]_4$ &$ [[67, 1 , d_e\geq18]]$&$ [[67, 1, 17]]$ &$ [[67, 1, 17]]$  \\
\makecell[l]{$16\times[[4,2,2;0]]$+$[[5,2,2;0]]$+ $[[3,2,2;1]]$} &  $[[17,1 ,9 ]]_4$ &$ [[68, 1 , d_e\geq18]]$&$ [[68, 1, 17]]$ &$ [[68, 0, 18]]$  \\
\makecell[l]{$14\times[[4,2,2;0]]$+$2\times[[5,2,2;0]]$+ $[[3,2,2;1]]$}&  $[[17,1 ,9 ]]_4$ &$ [[69, 1 , d_e\geq18]]$&$ [[69, 1, 17]]$ &$ [[69, 1, 17]]$  \\
\makecell[l]{$13\times[[4,2,2;0]]$+$3\times[[5,2,2;0]]$+ $[[3,2,2;1]]$} &  $[[17,1 ,9 ]]_4$ &$ [[70, 1 , d_e\geq18]]$&$ [[70, 1, 17]]$ &$ [[70, 0, 18]]$  \\
\makecell[l]{$9\times[[4,2,2;0]]$+$7\times[[5,2,2;0]]$+ $[[3,2,2;1]]$} &  $[[17,1 ,9 ]]_4$ &$ [[74, 1 , d_e\geq18]]$&$ [[74, 1, 17]]$ &$ [[74, 0, 18]]$  \\
$17\times[[4,2,2;0]]$  &  $[[17,3 ,8]]_4$ &$ [[68, 6 , d_e\geq16]]$&$ [[68, 6, 14]]$ &$ [[68, 6, 15]]$ \\
\makecell[l]{$16\times[[4,2,2;0]]$+$[[5,2,2;0]]$} &  $[[17,3 ,8]]_4$ &$ [[69, 6 , d_e\geq16]]$&$ [[69, 6, 14]]$ &$ [[69, 5, 15]]$\\
\makecell[l]{$15\times[[4,2,2;0]]$+$2\times[[5,2,2;0]]$} &  $[[17,3 ,8]]_4$ &$ [[70, 6 , d_e\geq16]]$&$ [[70, 6, 14]]$ &$ [[70, 6, 15]]$\\
\makecell[l]{$14\times[[4,2,2;0]]$+$3\times[[5,2,2;0]]$} &  $[[17,3 ,8]]_4$ &$ [[71, 6 , d_e\geq16]]$&$ [[71, 6, 14]]$ &$ [[71, 5, 16]]$\\
\makecell[l]{$12\times[[4,2,2;0]]$+$5\times[[5,2,2;0]]$} &  $[[17,3 ,8]]_4$ &$ [[73, 6 , d_e\geq16]]$&$ [[73, 6, 14]]$ &$ [[73, 5, 16]]$\\
\makecell[l]{$16\times[[4,2,2;0]]$+$[[3,2,2;1]]$} &  $[[17,3 ,8]]_4$ &$ [[67, 5 , d_e\geq16]]$&$ [[67, 5, 14]]$ &$ [[67, 5, 15]]$\\
\makecell[l]{$15\times[[4,2,2;0]]$+$2\times[[3,2,2;1]]$} &  $[[17,3 ,8]]_4$ &$ [[66, 4 , d_e\geq16]]$&$ [[66, 4, 14]]$ &$ [[66, 4, 15]]$\\
\makecell[l]{$15\times[[4,2,2;0]]$+$[[5,2,2;0]]$+ $[[3,2,2;1]]$} &  $[[17,3 ,8]]_4$ & $ [[68,5, d_e\geq16]]$ &   $ [[68,5 , 14]]$  &   $ [[68,4 , 15]]$   \\
\makecell[l]{$14\times[[4,2,2;0]]$+$2\times[[5,2,2;0]]$ $[[3,2,2;1]]$} &  $[[17,3 ,8]]_4$ & $ [[69,5, d_e\geq16]]$ &   $ [[69,5 , 14]]$  &   $ [[69,5, 15]]$   \\
\makecell[l]{$13\times[[4,2,2;0]]$+$3\times[[5,2,2;0]]$ $[[3,2,2;1]]$} &  $[[17,3 ,8]]_4$ & $ [[70,5, d_e\geq16]]$ &   $ [[70,5 , 14]]$  &   $ [[70,4 , 16]]$   \\
\makecell[l]{$14\times[[4,2,2;0]]$+$[[5,2,2;0]]$+ $2\times[[3,2,2;1]]$} &  $[[17,3 ,8]]_4$ & $ [[67,4, d_e\geq16]]$ &  $ [[67,4 , 14]]$   &  $ [[67,3 , 16]]$ \\
\makecell[l]{$13\times[[4,2,2;0]]$+$2\times[[5,2,2;0]]$+ $2\times[[3,2,2;1]]$} &  $[[17,3 ,8]]_4$ & $ [[68,4, d_e\geq16]]$ &  $ [[68,4 , 14]]$   &  $ [[68,4 , 15]]$ \\
\makecell[l]{$12\times[[4,2,2;0]]$+$3\times[[5,2,2;0]]$+ $2\times[[3,2,2;1]]$} &  $[[17,3 ,8]]_4$ & $ [[69,4, d_e\geq16]]$ &  $ [[69,4 , 14]]$   &  $ [[69,3, 16]]$ \\
$13\times[[10,2,4;0]]$  &  $[[13,1 ,7]]_4$ &$ [[130, 2 , d_e\geq28]]$&      $ \textbf{[[130, 2, 26]]}$  &      $ \textbf{[[130, 2, 27]]}$  \\
\makecell[l]{$12\times[[10,2,4;0]]$+$[[11,2,4;0]]$} &  $[[13,1 ,7]]_4$ & $ [[131, 2, d_e\geq28]]$& $\textbf{[[131, 2, 26]]} $& $\textbf{[[131, 1, 27]]} $  \\
\makecell[l]{$11\times[[10,2,4;0]]$+$2\times[[11,2,4;0]]$} &  $[[13,1 ,7]]_4$ & $ [[132, 2, d_e\geq28]]$& $\textbf{[[132, 2, 26]]} $& $\textbf{[[132, 2, 27]]} $  \\
\makecell[l]{$10\times[[10,2,4;0]]$+$3\times[[11,2,4;0]]$} &  $[[13,1 ,7]]_4$ & $ [[133, 2, d_e\geq28]]$& $\textbf{[[133, 2, 26]]} $& $\textbf{[[133, 1, 27]]} $  \\
\makecell[l]{$9\times[[10,2,4;0]]$+$4\times[[11,2,4;0]]$} &  $[[13,1 ,7]]_4$ & $ [[134, 2, d_e\geq28]]$& $ {[[134, 2, 26]]} $& $\textbf{[[134, 1, 27]]} $  \\
\makecell[l]{$8\times[[10,2,4;0]]$+$5\times[[11,2,4;0]]$} &  $[[13,1 ,7]]_4$ & $ [[135, 2, d_e\geq28]]$& $ {[[135, 2, 27]]} $& $\textbf{[[135, 2, 28]]} $  \\
\makecell[l]{$7\times[[10,2,4;0]]$+$6\times[[11,2,4;0]]$} &  $[[13,1 ,7]]_4$ & $ [[136, 2, d_e\geq28]]$& $ {[[136, 2, 27]]} $& $\textbf{[[136, 2, 28]]} $  \\
\makecell[l]{$6\times[[10,2,4;0]]$+$7\times[[11,2,4;0]]$} &  $[[13,1 ,7]]_4$ & $ [[137, 2, d_e\geq28]]$& $ {[[137, 2, 27]]} $& $\textbf{[[137, 1, 28]]} $  \\
$14\times[[10,2,4;0]]$  &  $[[14,2 ,7]]_4$ &$ [[140, 4 , d_e\geq28]]$&      $ \textbf{[[140, 4, 27]]}$&      $ \textbf{[[140, 4, 28]]}$   \\
\makecell[l]{$13\times[[10,2,4;0]]$+$[[11,2,4;0]]$} &  $[[14,2 ,7]]_4$ & $ [[141, 4, d_e\geq28]]$ &  $ \textbf{[[141, 4, 27]]} $ &  $ \textbf{[[141, 3, 28]]} $  \\
\makecell[l]{$12\times[[10,2,4;0]]$+$2\times[[11,2,4;0]]$} &  $[[14,2 ,7]]_4$ & $ [[142, 4, d_e\geq28]]$ &  $ \textbf{[[142, 4, 27]]} $ &  $ \textbf{[[142, 4, 28]]} $  \\

$15\times[[10,2,4;0]]$  &  $[[15,1 ,8]]_4$ &$ [[150, 2 , d_e\geq32]]$&     $ [[150, 2, 30]]$ &     $ \textbf{[[150, 2, 30]]}$  \\
\makecell[l]{$14\times[[10,2,4;0]]$+$[[11,2,4;0]]$} &  $[[15,1 ,8]]_4$ & $ [[151, 2, d_e\geq32]]$ &   $ [[151, 2, 30]] $ &
 $ \textbf{[[151, 1, 31]]} $ \\

\makecell[l]{$13\times[[10,2,4;0]]$+$2\times[[11,2,4;0]]$} &  $[[15,1 ,8]]_4$ & $ [[152, 2, d_e\geq32]]$ &   $ [[152, 2, 30]] $ &
 $ \textbf{[[152, 2, 31]]} $ \\
\makecell[l]{$12\times[[10,2,4;0]]$+$3\times[[11,2,4;0]]$} &  $[[15,1 ,8]]_4$ & $ [[153, 2, d_e\geq32]]$ &   $ [[153, 2, 30]] $ &
 $ \textbf{[[153, 1, 31]]} $ \\
\makecell[l]{$11\times[[10,2,4;0]]$+$4\times[[11,2,4;0]]$} &  $[[15,1 ,8]]_4$ & $ [[154, 2, d_e\geq32]]$ &   $ [[154, 2, 30]] $ &
 $ \textbf{[[154, 2, 31]]} $ \\
\makecell[l]{$10\times[[10,2,4;0]]$+$5\times[[11,2,4;0]]$} &  $[[15,1 ,8]]_4$ & $ [[155, 2, d_e\geq32]]$ &   $ [[155, 2, 30]] $ &
 $ \textbf{[[155, 1, 32]]} $ \\
\makecell[l]{$9\times[[10,2,4;0]]$+$6\times[[11,2,4;0]]$} &  $[[15,1 ,8]]_4$ & $ [[156, 2, d_e\geq32]]$ &   $ [[156, 2, 30]] $ &
 $ \textbf{[[156, 2, 32]]} $ \\
\makecell[l]{$8\times[[10,2,4;0]]$+$7\times[[11,2,4;0]]$} &  $[[15,1 ,8]]_4$ & $ [[157, 2, d_e\geq32]]$ &   $ \textbf{[[157, 2, 31]]} $ &
 $ \textbf{[[157, 1, 32]]} $ \\
\makecell[l]{$7\times[[10,2,4;0]]$+$8\times[[11,2,4;0]]$} &  $[[15,1 ,8]]_4$ & $ [[158, 2, d_e\geq32]]$ &   $ \textbf{[[158, 2, 31]]} $ &
 $ \textbf{[[158, 2, 32]]} $ \\
\makecell[l]{$6\times[[10,2,4;0]]$+$9\times[[11,2,4;0]]$} &  $[[15,1 ,8]]_4$ & $ [[159, 2, d_e\geq32]]$ &   $ \textbf{[[159, 2, 31]]} $ &
 $ \textbf{[[159, 1, 32]]} $ \\
\makecell[l]{$5\times[[10,2,4;0]]$+$10\times[[11,2,4;0]]$} &  $[[15,1 ,8]]_4$ & $ [[160, 2, d_e\geq32]]$ &   $ \textbf{[[160, 2, 31]]} $ &
 $ \textbf{[[160, 2, 32]]} $ \\
\hline
\end{tabular}
\end{center}
\label{tables1}
\end{table*}

\begin{table*}\tiny
 \caption{\scriptsize EACQCs with better parameters  than the best known   QECCs and EAQECCs   of the same length and net transmission (cont.). {The notation of the inner codes, for example, ``{$14\times[[10,2,4;0]]$+$2\times[[11,2,4;0]]$}'',  means
 that the inner codes are a mix of fourteen   $[[10,2,4;0]]$ codes and two $[[11,2,4;0]]$ codes.  The  outer codes are  EAQMDS codes in Lemma S2. The parameters $ k^*_2$, $k^*_e$ and $K^*_e$ stand for  the net transmissions of the $[[n_2,k^*_2,d_2]]_{2^{k_1}}$ outer code, the $ [[n_e,k^*_e,d_e ]]$ EACQC  and the $ [[N_e,K^*_e,D_e ]]$ EAQECC, respectively.  EAQECCs in the last column are obtained from the best known quaternary  codes in Ref.~\cite{Grassl:codetables}.  For the QECCs and the EAQECCs  missing explicit constructions   in  Ref.~\cite{Grassl:codetables}, their  parameters are in \textbf{bold} type. }}
\begin{center}
\begin{tabular}{lllll}
\makecell[l]{$[[n_1,k_1,d_1;c_1]]$\\ Inner Codes}     &  \makecell[l]{$[[n_2,k^*_2,d_2]]_{2^{k_1}}$ \\Outer Codes}   &\makecell[l]{$[[n_e,k^*_e,d_e ]]$\\EACQCs }&    \makecell[l]{$[[n ,k ,d ]] $ \\ QECCs in  Ref.~\cite{Grassl:codetables}}&\makecell[l]{$[[N_e,K^*_e,D_e ]]$\\EAQECCs  from  Ref.~\cite{Grassl:codetables}}\\
\hline
$16\times[[10,2,4;0]]$  &  $[[16,2 ,8]]_4$ &$ [[160, 4 , d_e\geq32]]$&     $ \textbf{[[160, 4, 31]]}$  &     $ \textbf{[[160, 4, 32]]}$   \\

\makecell[l]{$15\times[[10,2,4;0]]$+$[[11,2,4;0]]$}  &  $[[16,2 ,8]]_4$ & $ [[161,4,d_e\geq32]]$ & $\textbf{[[161, 4, 31]]} $
& $\textbf{[[161, 3, 32]]} $   \\

\makecell[l]{$14\times[[10,2,4;0]]$+$2\times[[11,2,4;0]]$}  &  $[[16,2 ,8]]_4$ & $ [[162,4,d_e\geq32]]$ & $\textbf{[[162, 4, 31]]} $
& $\textbf{[[162, 4, 32]]} $   \\

\makecell[l]{$15\times[[10,2,4;0]]$+$[[8,2,4;2]]$}  &  $[[16,2 ,8]]_4$ &$ [[158, 2, d_e\geq32]]$&      $ \textbf{[[158, 2, 31]]}$
&      $ \textbf{[[158, 2, 32]]}$  \\

\makecell[l]{$14\times[[10,2,4;0]]$+$[[11,2,4;0]]$+$[[8,2,4;2]]$}  &  $[[16,2 ,8]]_4$ & $ [[159, 2, d_e\geq32]]$ &   $ \textbf{[[159, 2, 31]]} $ &   $ \textbf{[[159, 1, 32]]} $  \\

\makecell[l]{$13\times[[10,2,4;0]]$+$2\times[[11,2,4;0]]$+$[[8,2,4;2]]$}  &  $[[16,2 ,8]]_4$ & $ [[160, 2, d_e\geq32]]$ &   $ \textbf{[[160, 2, 31]]} $ &   $ \textbf{[[160, 2, 32]]} $  \\

$17\times[[10,2,4;0]]$  &  $[[17,1,9]]_4$ &$ [[170, 2 , d_e\geq36]]$ &    $ \textbf{[[170, 2, 33]]}$ &    $ \textbf{[[170, 2, 34]]}$  \\

\makecell[l]{$16\times[[10,2,4;0]]$+$[[11,2,4;0]]$} &  $[[17,1 ,9]]_4$ &$ [[171, 2 , d_e\geq36]]$&    $ \textbf{[[171, 2, 33]]}$ &    $ \textbf{[[171, 1, 35]]}$  \\

\makecell[l]{$15\times[[10,2,4;0]]$+$2\times[[11,2,4;0]]$} &  $[[17,1 ,9]]_4$ &$ [[172, 2 , d_e\geq36]]$&    $ \textbf{[[172, 2, 34]]}$ &    $ \textbf{[[172, 2, 35]]}$  \\

\makecell[l]{$14\times[[10,2,4;0]]$+$3\times[[11,2,4;0]]$} &  $[[17,1 ,9]]_4$ &$ [[173, 2 , d_e\geq36]]$&    $ \textbf{[[173, 2, 34]]}$ &    $ \textbf{[[173, 1, 35]]}$  \\

\makecell[l]{$13\times[[10,2,4;0]]$+$4\times[[11,2,4;0]]$} &  $[[17,1 ,9]]_4$ &$ [[174, 2 , d_e\geq36]]$&    $ \textbf{[[174, 2, 34]]}$ &    $ \textbf{[[174, 2, 35]]}$  \\

\makecell[l]{$12\times[[10,2,4;0]]$+$5\times[[11,2,4;0]]$} &  $[[17,1 ,9]]_4$ &$ [[175, 2 , d_e\geq36]]$&    $ \textbf{[[175, 2, 34]]}$ &    $ \textbf{[[175, 1, 35]]}$  \\

\makecell[l]{$11\times[[10,2,4;0]]$+$6\times[[11,2,4;0]]$} &  $[[17,1 ,9]]_4$ &$ [[176, 2 , d_e\geq36]]$&    $ \textbf{[[176, 2, 34]]}$ &    $ \textbf{[[176, 2, 35]]}$  \\

\makecell[l]{$10\times[[10,2,4;0]]$+$7\times[[11,2,4;0]]$} &  $[[17,1 ,9]]_4$ &$ [[177, 2 , d_e\geq36]]$&    $ \textbf{[[177, 2, 34]]}$ &    $ \textbf{[[177, 1, 36]]}$  \\

\makecell[l]{$9\times[[10,2,4;0]]$+$8\times[[11,2,4;0]]$} &  $[[17,1 ,9]]_4$ &$ [[178, 2 , d_e\geq36]]$&    $ \textbf{[[178, 2, 35]]}$ &    $ \textbf{[[178, 2, 36]]}$  \\

\makecell[l]{$8\times[[10,2,4;0]]$+$9\times[[11,2,4;0]]$} &  $[[17,1 ,9]]_4$ &$ [[179, 2 , d_e\geq36]]$&    $ \textbf{[[179, 2, 35]]}$ &    $ \textbf{[[179, 1, 36]]}$  \\

\makecell[l]{$7\times[[10,2,4;0]]$+$10\times[[11,2,4;0]]$} &  $[[17,1 ,9]]_4$ &$ [[180, 2 , d_e\geq36]]$&    $ \textbf{[[180, 2, 35]]}$ &    $ \textbf{[[180, 2, 36]]}$  \\

\makecell[l]{$6\times[[10,2,4;0]]$+$11\times[[11,2,4;0]]$} &  $[[17,1 ,9]]_4$ &$ [[182, 2 , d_e\geq36]]$&    $ \textbf{[[182, 2, 35]]}$ &    $ \textbf{[[182, 2, 36]]}$  \\

   $15\times[[16,2,6;0]]$  &  $[[15,1 ,8]]_4$ &$ [[240, 2 , d_e\geq48]]$&     $ \textbf{[[240, 2, 46]]}$  &     $ \textbf{[[240, 2, 48]]}$   \\

\makecell[l]{$14\times[[16,2,6;0]]$+$[[17,2,6;0]]$} &  $[[15,1 ,8]]_4$ &$ [[241, 2 , d_e\geq48]]$&    $ \textbf{[[241, 2, 47]]}$ &    $ \textbf{[[241, 1, 48]]}$  \\

\makecell[l]{$13\times[[16,2,6;0]]$+$2\times[[17,2,6;0]]$} &  $[[15,1 ,8]]_4$ &$ [[242, 2 , d_e\geq48]]$&    $ \textbf{[[242, 2, 47]]}$ &    $ \textbf{[[242, 2, 48]]}$  \\

\makecell[l]{$12\times[[16,2,6;0]]$+$3\times[[17,2,6;0]]$} &  $[[15,1 ,8]]_4$ &$ [[243, 2 , d_e\geq48]]$&    $ \textbf{[[243, 2, 47]]}$ &    $ \textbf{[[243, 1, 48]]}$  \\

\makecell[l]{$11\times[[16,2,6;0]]$+$4\times[[17,2,6;0]]$} &  $[[15,1 ,8]]_4$ &$ [[244, 2 , d_e\geq48]]$&    $ \textbf{[[244, 2, 47]]}$ &    $ \textbf{[[244, 2, 48]]}$  \\

 $9\times[[28,2,10;0]]$  &  $[[9,1 ,5]]_4$ &$ [[252, 2 , d_e\geq50]]$&     $ \textbf{[[252, 2, 49]]}$  &     $ \textbf{[[252, 2, 50]]}$   \\

\makecell[l]{$8\times[[28,2,10;0]]$+$[[29,2,10;0]]$} &  $[[9,1 ,5]]_4$ &$ [[253, 2 , d_e\geq50]]$&    $ \textbf{[[253, 2, 49]]}$ &    $ \textbf{[[253, 1, 50]]}$  \\

\makecell[l]{$7\times[[28,2,10;0]]$+$2\times[[29,2,10;0]]$} &  $[[9,1 ,5]]_4$ &$ [[254, 2 , d_e\geq50]]$&    $ \textbf{[[254, 2, 49]]}$ &    $ \textbf{[[254, 2, 50]]}$  \\
\hline
\end{tabular}
\end{center}
\label{tables2}
\end{table*}
\newpage
\begin{table*}
\tiny
\caption{\scriptsize EACQCs with better parameters  than  the best known QECCs  and EAQECCs   of the same length and net transmission. {The notation of the inner codes, for example, ``$2\times[[29,1,11]]$+$ [[30,1,11]]$'',  means
 that the inner codes are a mix of two $[[29,1,11]]$ codes and a $[[30,1,11]]$ code.   EAQECCs in the last column are obtained from the best known quaternary  codes in Ref.~\cite{Grassl:codetables}, and the number of ebits is computed with MAGMA. }}
 \begin{center}
\begin{tabular}{lllll}
  \makecell[l]{Inner Codes }    & \makecell[l]{Outer  Codes}   & EACQCs  &  \makecell[l]{QECCs in Ref. \cite{Grassl:codetables}} &\makecell[l]{EAQECCs with MAGMA} \\
\hline
$3\times[[5,1,3]]$  & $[[3,2,2;1]]$ &$ [[15, 2, 6;1]]$& $ [[15, 1, 5]]$ & $ [[15, 8, 6;7]]$ \\
$3\times[[17,1,7]]$  &  $[[3,2,2;1]]$ &$ [[51, 2, 14;1]]$&$ [[51, 1, 13]]$ &$ [[51, 2, 14;1]]$\\
$3\times[[25,1,9]]$  &  $[[3,2,2;1]]$ &$ [[75, 2, 18;1]]$&$ [[75, 1, 17]]$ &$ [[75, 37, 18;36]]$ \\
$3\times[[29,1,11]]$  &  $[[3,2,2;1]]$ &$ [[87, 2, 22;1]]$&$ [[87, 1, 21]]$  &$ [[87, 43, 21;42]]$   \\
\makecell[l]{$2\times[[29,1,11]]$+$ [[30,1,11]]$} &  $[[3,2,2;1]]$ &$ [[88, 2, 22;1]]$&$ [[88, 1, 21]]$ &$ [[88, 27, 21;27]]$ \\
\makecell[l]{$2\times[[29,1,11]]$+$ [[31,1,11]]$} &  $[[3,2,2;1]]$ &$ [[89, 2, 22;1]]$&$ [[89, 1, 21]]$ &$ [[89, 29, 21;28]]$ \\
\makecell[l]{$2\times[[29,1,11]]$+$ [[32,1,11]]$} &  $[[3,2,2;1]]$ &$ [[90, 2, 22;1]]$&$ [[90, 1, 21]]$ &$ [[90, 31, 21;31]]$ \\
\makecell[l]{$2\times[[29,1,11]]$+$ [[33,1,11]]$} &  $[[3,2,2;1]]$ &$ [[91, 2, 22;1]]$&$ [[91, 1, 21]]$ &$ [[91, 30, 21;29]]$ \\
\makecell[l]{$2\times[[29,1,11]]$+$ [[34,1,11]]$} &  $[[3,2,2;1]]$ &$ [[92, 2, 22;1]]$&$ [[92, 1, 21]]$ &$ [[92, 0, 22;0]]$ \\
\makecell[l]{$2\times[[29,1,11]]$+$ [[35,1,11]]$} &  $[[3,2,2;1]]$ &$ [[93, 2, 22;1]]$&$ [[93, 1, 21]]$ &$ [[93, 32, 21;31]]$ \\
\makecell[l]{$2\times[[29,1,11]]$+$ [[36,1,11]]$} &  $[[3,2,2;1]]$ &$ [[94, 2, 22;1]]$&$ [[94, 1, 21]]$ &$ [[94, 2, 22;2]]$ \\
\makecell[l]{$2\times[[29,1,11]]$+$ [[37,1,11]]$} &  $[[3,2,2;1]]$ &$ [[95, 2, 22;1]]$&$ [[95, 1, 21]]$ &$ [[95, 2, 22;1]]$ \\
\makecell[l]{$2\times[[29,1,11]]$+$ [[39,1,11]]$} &  $[[3,2,2;1]]$ &$ [[97, 2, 22;1]]$&$ [[97, 1, 21]]$ &$ [[97, 45, 22;44]]$ \\
\makecell[l]{$2\times[[29,1,11]]$+$ [[40,1,11]]$} &  $[[3,2,2;1]]$ &$ [[98, 2, 22;1]]$&$ [[98, 1, 21]]$ &$ [[98, 0, 22;0]]$ \\
\makecell[l]{$2\times[[29,1,11]]$+$ [[41,1,11]]$} &  $[[3,2,2;1]]$ &$ [[99, 2, 22;1]]$&$ [[99, 1, 21]]$ &$ [[99, 32, 21;31]]$ \\
\makecell[l]{$2\times[[30,1,11]]$+$ [[40,1,11]]$} &  $[[3,2,2;1]]$ &$ [[100, 2, 22;1]]$&$ [[100, 1, 21]]$ &$ [[100, 0, 22;0]]$ \\
\makecell[l]{$2\times[[30,1,11]]$+$ [[41,1,11]]$} &  $[[3,2,2;1]]$ &$ [[101, 2, 22;1]]$&$ [[101, 1, 21]]$ &$ [[101, 32, 21;31]]$ \\
\makecell[l]{$2\times[[31,1,11]]$+$ [[40,1,11]]$} &  $[[3,2,2;1]]$ &$ [[102, 2, 22;1]]$&$ [[102, 1, 21]]$ &$ [[102, 0, 22;0]]$ \\
\makecell[l]{$2\times[[31,1,11]]$+$ [[41,1,11]]$} &  $[[3,2,2;1]]$ &$ [[103, 2, 22;1]]$&$ [[103, 1, 21]]$ &$ [[103, 4, 22;3]]$ \\
\hline
\end{tabular}
 \end{center}
 \label{tables3}
\end{table*}

\begin{table*}\tiny
 \caption{\scriptsize EACQCs  with better parameters  than the best known     QECCs and EAQECCs   of the same length and net transmission. {The outer codes are EAQMDS codes in Ref.~\cite{fan2016constructions}. EAQECCs in the last column are obtained from the best known quaternary  codes in Ref.~\cite{Grassl:codetables}, and the number of ebits is computed with MAGMA.  For the QECCs   {missing explicit constructions} in  Ref.~\cite{Grassl:codetables}, their  parameters are in \textbf{bold} type. Since we cannot compute the number of ebits for the best known quaternary codes  {missing explicit constructions} in Ref.~\cite{Grassl:codetables}, we leave some EAQECCs empty in the last column.}}
\begin{center}
\begin{tabular}{lllll}
 Inner Codes     &  Outer Codes   & EACQCs &   QECCs in Ref.~\cite{Grassl:codetables}&EAQECCs with MAGMA \\
\hline
$17\times[[4,2,2]]$  &  $[[17,4,8;1]]_4$ &$ [[68, 8, 16;2]]$&$ [[68, 6, 14]]$ &$ [[68, 19, 15;13]]$ \\

 {$16\times[[4,2,2]]$+$[[5,2,2]]$ } &  $[[17,4,8;1]]_4$ & {$ [[69, 8, 16;2]]$ }& $ [[69, 6, 14]] $& $ [[69, 19, 15;14]] $  \\

 {$15\times[[4,2,2]]$+$2\times[[5,2,2]]$ } &  $[[17,4,8;1]]_4$ & {$ [[70, 8, 16;2]]$ }& $ [[70, 6, 14]] $& $ [[70, 20, 15;14]] $  \\

 {$14\times[[4,2,2]]$+$3\times[[5,2,2]]$ } &  $[[17,4,8;1]]_4$ & {$ [[71, 8, 16;2]]$ }& $ [[71, 6, 14]] $& $ [[71, 32, 16;27]] $  \\

 {$13\times[[4,2,2]]$+$4\times[[5,2,2]]$ } &  $[[17,4,8;1]]_4$ & {$ [[72, 8, 16;2]]$ }& $ [[72, 6, 14]] $& $ [[72, 32, 16;26]] $  \\

  {$12\times[[4,2,2]]$+$5\times[[5,2,2]]$ } &  $[[17,4,8;1]]_4$ & {$ [[73, 8, 16;2]]$ }& $ [[73, 6, 14]] $& $ [[73, 25, 16;20]] $  \\

  {$11\times[[4,2,2]]$+$6\times[[5,2,2]]$ } &  $[[17,4,8;1]]_4$ & {$ [[74, 8, 16;2]]$ }& $ [[74, 6, 15]] $& $ [[74, 32, 16;27]] $  \\
$17\times[[10,2,4;0]]$   &  $[[17,4,8;1]]_4$ &$ [[170, 8, 32;2]]$&   $ \textbf{[[170, 6, 32]]}$ &-----\\
  \makecell[l]{$(17-t)\times[[10,2,4]]$+\\$t\times[[11,2,4]]$, $1\leq t\leq 2$} &  $[[17,4,8;1]]_4$ & \makecell[l]{$ [[170+t, 8, 32;2]]$,\\ $1\leq t\leq 2$}& \makecell[l]{$ \textbf{[[170+t, 6, 32 ]]}$,\\ $1\leq t\leq 2$}&----- \\
   \makecell[l]{$(17-t)\times[[10,2,4]]$+\\$t\times[[8,2,4;2]]$, $1\leq t\leq 2$} &  $[[17,4,8;1]]_4$ & \makecell[l]{$ [[170-2t, 8, 32;2+2t]]$, \\$1\leq t\leq 2$}& \makecell[l]{$ \textbf{[[170-2i, 6-2t, 32 ]]}$, \\$1\leq t\leq 2$} &-----\\
   \hline
\end{tabular}
\end{center}
\label{tables4}
\end{table*}
\clearpage
\section{The Entanglement Fidelity of EACQCs}
\label{App3}
 We use the entanglement  fidelity (EF) as the figure of merit for the performance of different quantum codes \cite{PhysRevA.56.131}.
Let  $\mathbb{H}$ be a finite dimensional Hilbert space.
Let $|\varphi\rangle\in\mathbb{H}\otimes \mathbb{H}_R $ be a purification of  a mixed state $ \rho=\textrm{Tr}_{\mathbb{H}_R}|\varphi\rangle\langle \varphi|$, where $\mathbb{H}_R $ is a reference system.  The entanglement fidelity  of $ \rho$ and $\Upsilon$ is defined as
\begin{equation}
F_e(\rho,\Upsilon)=\langle \varphi|(I_{\mathbb{H}_R}\otimes \Upsilon)(|\varphi\rangle\langle \varphi|)|\varphi\rangle,
\end{equation}
where $I_{\mathbb{H}_R}$ is the identity operator and $\Upsilon$ is a quantum map. Suppose that we can write the quantum
map $\Upsilon$ in terms of Kraus  operators, i.e., $\Upsilon (\rho)=\sum_iA_i\rho A_i^\dag$, where $A_i^\dag A_i=I$, then the EF can be expressed by  a useful computational
formula as follows:
 \begin{equation}
F_e(\rho,\Upsilon)= \sum_i|\textrm{Tr}(\rho A_i)|^2.
\end{equation}

 As a special case of   EF, the   \emph{channel fidelity}     \cite{PhysRevA.56.131,lai2012entanglement}  is defined as
\begin{equation}
F_c(\rho)= \frac{1}{(\dim \mathbb{H})^2}\sum_i|\textrm{Tr}(  A_i)|^2.
\end{equation}
For the depolarizing channel,  the channel fidelity is equal to the  probability of correctable errors after quantum error correction (QEC) and recovery \cite{lai2012entanglement}.
Denote the encoding and recovery operations in a QEC process by $\mathcal{E}$ and $\mathcal{R}$, respectively. For  a single qubit state $\rho_0$,    the encoding $\mathcal{E}$ takes it to an encoded state $\rho(0)$, i.e., $\mathcal{E}:\rho_0\rightarrow \rho(0)$. The state $\rho(0)$ is sent through the quantum channel $\Upsilon$ and the received state is
$\rho(t)=\Upsilon(\rho(0))$. At the receiver, $\rho(t) $ is recovered by the decoding and recovery operation $\mathcal{R}$ and the final logical state is   $\rho_f=\mathcal{R}(\rho(t))$. For the entire  QEC process, the operation
\begin{equation}
 \mathcal{W}=\mathcal{R}\circ \Upsilon   \circ \mathcal{E}: \rho_0\rightarrow \rho_f
\end{equation}
is called the effective channel, which characterizes the effective dynamics of the encoded information arising from the physical dynamics of $\Upsilon$ \cite{rahn2002exact}.

Suppose that we can write the quantum  channel $\Upsilon$ on $N$ physical qubits as independent noise $\Upsilon^{(1)}$ on each single qubit, i.e.,
\begin{equation}
\Upsilon=\Upsilon^{(1)}\otimes\cdots\otimes\Upsilon^{(1)}=\Upsilon^{(1)\otimes N}.
\end{equation}
Let $Q $ be an $N$ physical qubit QECC with encoding
operation and recovery operation given by $ \mathcal{E}$ and $ \mathcal{R}$, respectively. Define the following \emph{coding map}:
\begin{equation}\label{codingmap}
\Omega^C: \Upsilon^{(1)}_0 \rightarrow \Upsilon^{(1)}_f=\mathcal{R} \circ     \Upsilon^{(1)\otimes N}_0 \circ \mathcal{E},
\end{equation}
 which takes
the single qubit noise $\Upsilon^{(1)}_0$ to the effective channel $\Upsilon^{(1)}_f$.
 If $\Upsilon^{(1)}_0$ ia a depolarizing channel and $Q$ is a stabilizer code, then $\Omega^C$ takes $\Upsilon^{(1)}_0$ to a depolarizing channel $\Upsilon^{(1)}_f$.

The channel fidelity of the $[[5,1,3]]$ stabilizer code over the depolarizing channel    \cite{lai2012entanglement,rahn2002exact} is given by
$
F_c^{[[5,1,3]]}=1-{45}p^2/8+{75}p^3/{8}-{45}p^4/{8}+{9}p^5/{8},
$
where $p$ is the depolarizing probability.  From \cite{rahn2002exact}, we known that the $[[5,1,3]]$ stabilizer code takes a depolarizing channel $\Upsilon^{(1)}_0 $ to a depolarizing channel $\Upsilon^{(1)}_f$ under the coding map  $\Omega^C$. The  error  probability of  $\Upsilon^{(1)}_f$ is equal to  $1-F_c^{[[5,1,3]]}$. According to \cite{cafaro2010quantum},  the entanglement fidelity of the $[[5,1,3]]$ code   is given by
$
F_e^{[[5,1,3]]}=1   -   10p^2 + 20p^3 - 15p^4+4p^5.
$
 Then we can derive the entanglement  fidelity of the $\mathcal{Q}_C=[[25,1,9]]$ CQC  as follows:
\begin{equation}
F_e^{[[25,1,9]]}=1   -   10p_C^2 + 20p_C^3 - 15p_C^4+4p_C^5,
\end{equation}
where $p_C= 1-F_c^{[[5,1,3]]}$.

Denote  the communication channel between Alice and Bob by $N_A$. Denote  the channel model of storing  Bob's ebits by $N_B$. Suppose that both $N_A$ and $N_B$ are
depolarizing channels, i.e.,
\begin{equation}
   \phi\mapsto \left(1-\frac{3p}{4}\right) \phi+\frac{p}{4}X\phi {X}+\frac{p}{4}Y\phi  {Y}+ \frac{p}{4}Z\phi {Z},
\end{equation}
where $\phi$ is a single quantum state, $0\leq p\leq 1$ is the depolarizing probability, and $\{X,Y,Z\}$ are the Pauli operators.
   For a $\mathscr{Q}_e=[[n,k,d;c]]$ EAQECC with $n$ physical qubits and $c$ ebits, denote
\begin{eqnarray}
\mu_i&=&\left(1-\frac{3p_a}{4} \right)^{n-i}\left(\frac{p_a}{4}\right)^i,\ \textrm{and}\\
\nu_j&=&\left(1-\frac{3p_b}{4}\right)^{c-j}\left(\frac{p_b}{4}\right)^j,
 \end{eqnarray}
where $0\leq i\leq n$, $0\leq j\leq c$, and, $p_a$ and $p_b$ are the depolarizing probabilities
of $N_A$ and $N_B$, respectively.

It is known that the channel fidelity of Bowen's $[[3,1,3;2]]$ EAQECC is given by
\begin{eqnarray}
F_c^{[[3,1,3;2]]}=\mu_0\nu_0 + 9\mu_1\nu_0 + 6\mu_3\nu_0 + 6\mu_0\nu_1 +  36\mu_2\nu_1 + 54\mu_3\nu_1
+18\mu_1\nu_2 + 81\mu_2\nu_2 + 45\mu_3\nu_2.
\end{eqnarray}
We  use Bowen's $[[3,1,3;2]]$ EAQECC as the inner code, and use  the $[[5,1,3]]$ stabilizer code as the outer code to construct a  $\mathscr{Q}_B=[[15,1,9;10]]^B$ EACQC. The inner $[[3,1,3;2]]$ EAQECC  takes a depolarizing channel to a  depolarizing channel under the coding map
$\Omega^C$  in  Eq.~[\ref{codingmap}]. Then the  entanglement fidelity of $\mathscr{Q}_B$ is given by
\begin{equation}
F_e^{[[15,1,9;10]]^B}=  1 - 10p_B^2 + 20p_B^3 - 15p_B^4+4p_B^5,
\end{equation}
where $p_B= 1-F_c^{[[3,1,3;2]]} .$

Alternately, we  use the $[[3,1,3;2]]$ EA repetition code as the inner code, and use  the $[[5,1,3]]$ stabilizer code as the outer code to construct another $\mathscr{Q}_R=[[15,1,9;10]]^R$ EACQC. It is known that the channel fidelity of the $[[3,1,3;2]]$ EA repetition code   is given  by
\begin{eqnarray}
F_c^{{[[3,1,3;2]]}'}=\mu_0\nu_0 + 9\mu_1\nu_0 + 6\mu_2\nu_0 + 18\mu_1\nu_1 +  38\mu_2\nu_1 + 40\mu_3\nu_1
+18\mu_1\nu_2 + 55\mu_2\nu_2 + 71\mu_3\nu_2.
\end{eqnarray}
Then the entanglement  fidelity of $\mathscr{Q}_R$ is given by
\begin{equation}
F_e^{{[[15,1,9;10]]}^R}=1 - 10p_R^2 + 20p_R^3 - 15p_R^4+4p_R^5,
\end{equation}
where $p_R= 1-F_c^{{[[3,1,3;2]]}'} $.

\end{document}